\documentclass{article}
\usepackage[a4paper]{geometry}
\usepackage[T1]{fontenc}
\usepackage[utf8]{inputenc}
\usepackage{lmodern}
\usepackage{tikz}
\usepackage{subcaption}

\usepackage{shuffle}
\usepackage{booktabs}
\usetikzlibrary{calc,patterns}

\usepackage{amsmath,amssymb}
\usepackage{amsthm}

\usepackage{xcolor}

\definecolor{col1}{RGB}{100,143,255}
\definecolor{col2}{RGB}{120, 94, 240}
\definecolor{col3}{RGB}{254,97,0}
\definecolor{col4}{RGB}{220, 38, 127}
\definecolor{col5}{RGB}{255, 176, 0}

\usepackage{hyperref}

\newtheorem{theorem}{Theorem}%
\newtheorem{lemma}[theorem]{Lemma}
\newtheorem{corollary}[theorem]{Corollary}
\newtheorem{proposition}[theorem]{Proposition}
\newtheorem{definition}[theorem]{Definition}

\newtheorem{algorithm}[theorem]{Algorithm}

\theoremstyle{remark}

\newcommand{\bb}[1]{{\boldsymbol{#1}}}

\newcommand{\tr}{\mathrm{tr}}

\newcommand{\R}{\mathbb{R}}
\newcommand{\Q}{\mathbb{Q}}

\newcommand{\HE}{{\mathcal{H}}}

\newcommand{\MG}{\mathcal{MG}}

\DeclareMathOperator{\Tr}{Tr}
\newcommand{\dd}{\text{d}}
\newcommand{\xiz}{z}

\newcommand{\bigO}{\mathcal{O}}
\newcommand{\smallO}{o}

\DeclareMathOperator{\id}{id}
\DeclareMathOperator{\Aut}{Aut}

\hypersetup{
  pdfauthor={Michael Borinsky},
  pdftitle={Tropicalized quantum field theory and global tropical sampling},
  pdfsubject={},
  pdfkeywords={},
  linkcolor  = col2,
  citecolor  = col4,
  urlcolor   = col5,
  colorlinks = true,
}
\begin{document}

\title{Tropicalized quantum field theory \\and global tropical sampling}
\author{Michael Borinsky%
\footnote{Perimeter Institute, 31 Caroline St N, Waterloo, ON N2L 2Y5, Canada
}
}
\date{}
\maketitle

\begin{abstract}
We explain how to tropicalize scalar quantum field theory and show that tropicalized massive scalar quantum field theory is exactly solvable. 
This exact solution manifests as a non-linear recursion equation fulfilled by the expansion coefficients of the quantum effective action.
Geometrically, this recursion computes specific volumes of moduli spaces of metric graphs and is thereby analogous to Mirzakhani's volume recursions on the moduli space of curves.
Building on this exact solution, we construct an algorithm that samples points from the moduli space of graphs approximately proportional to their perturbative contribution. Remarkably, this algorithm requires only polynomial time and memory, suggesting that perturbative quantum field theory computations lie in the polynomial-time complexity class, while all known algorithms for evaluating individual Feynman integrals are exponential in time and memory.
To demonstrate the capabilities of the algorithm, we evaluate the primitive contribution to the $\phi^4$ beta function at 50 loops with a proof-of-concept implementation.
\end{abstract}
\section{Introduction}
\subsection{Physical motivation}

Quantum field theory is regarded as one of the most accurate theories in science, with predictions that agree with experiments to an extraordinary degree of precision. Yet it remains very difficult to obtain concrete predictions from it. One of the most severe bottlenecks is the need to sum over all Feynman integrals of a given loop order to reach a required accuracy within a perturbation-theoretic workflow. As the number of these integrals and the associated graphs grows factorially, the amount of work necessary to make a quantum field theory prediction at a specified accuracy also grows at least factorially. In contrast, the time, energy, and resources needed to measure such a physical quantity typically grow at most polynomially with the desired accuracy.

From this physical perspective, this article is motivated by the belief that this is a highly unsatisfying situation: a well-formulated law of nature should not only allow for concrete predictions of a physical phenomenon but also enable the observer to make these predictions with lower overall effort than performing the experiment and measuring the outcome.

\subsection{Overview of the results}

This article addresses this conceptual deficiency of quantum field theory in two steps: 

First, we define a deformation of massive scalar quantum field theory inspired by tropical geometry, which we call \emph{tropicalized quantum field theory}.
Our first main result (Theorem~\ref{thm:pde}) is that this theory is \emph{exactly solvable}:  
its quantum effective action satisfies a non-linear partial differential equation, the \emph{tropical loop equation}, that determines all perturbative coefficients recursively. This exact solution holds in any spacetime dimension; however, in the critical dimension, this equation simplifies to an ordinary differential equation.

Second, using this solvability, we design a sampling algorithm that effectively approximates the original perturbative expansion without enumerating individual Feynman integrals.
The algorithm samples points from the moduli space of graphs.
The probability density with which the points are sampled is the tropicalization of the original Schwinger parametric Feynman integrand. The perturbative expansion coefficient of the initial, non-deformed quantum field theory can be computed via a Monte Carlo procedure, with convergence guaranteed by tropical algebra. Crucially, this new sampling algorithm has only \emph{polynomially growing runtime and memory requirements}. The existence of this fast algorithm implies, for instance, that the evaluation of a specific Feynman integral contributing to the perturbative expansion becomes quickly more computationally demanding  (with currently available methods) than estimating the global sum over all Feynman integrals at once.

As a proof of concept, we compute perturbative coefficients for the $3$-point correlation function in massive $\phi^3$ theory in three dimensions up to 20 loops. For $\phi^4$ theory, we compute some coefficients of the primitive $\beta$-function up to 50 loops.

\subsection{Relation to previous works}
Over the last few decades, there have been many pioneering efforts to simplify and make quantities such as scattering amplitudes in quantum field theory both conceptually and computationally more manageable. Examples include methods that sum all Feynman diagrams of a particular class at once (see, e.g.,~\cite{Parke:1986gb,Bern:1994zx}), as well as the discovery of new, elegant formulations that capture large amounts of combinatorial data  (see, e.g.,~\cite{Britto:2005fq,Arkani-Hamed:2010zjl,Cachazo:2013hca}).

Over the last decade, specifically, tropical and convex geometric methods became established tools in fundamental physics (see, e.g.,~\cite{kaneko2010geometric,arkani2014amplituhedron,arkani2017positive,Arkani-Hamed:2019mrd}). 
Panzer introduced the tropicalization of individual Feynman integrals relevant for this work: the \emph{Hepp bound}~\cite{Panzer:2019yxl}. The Hepp bound and more general Feynman integral tropicalizations were used by the author in \cite{Borinsky:2020rqs} to give a sampling algorithm for individual Feynman integrals whose runtime scales better than previous algorithms (e.g.,~\cite{Binoth:2003ak}), but still exponentially with the number of edges of the Feynman graph.

    Further, our approach is inspired by the surfaceology program \cite{arkani2023all,Arkani-Hamed:2024vna,Arkani-Hamed:2024pzc,De:2024wsy,Salvatori:2025oib}. In particular, our perspective resonates with Salvatori's recent work \cite{Salvatori:2025oib}, which puts forward a sampling algorithm for surfaces through Hepp-bound-related recursions. 
 This algorithm allows one to directly compute entire perturbative expansion coefficients of certain quantum field theories. %
The second part of the present article translates some of these ideas from surfaceology into a purely graph-based setting, without resorting to surfaces at all. The global sampling algorithm presented in Section~\ref{sec:algo} achieves the same goal of providing a way to compute a full perturbative coefficient at once without resorting to a decomposition into individual Feynman integrals. 
Its runtime and memory requirements, which only grow polynomially with the loop order and multiplicity, appear to improve significantly over those of the algorithm in \cite{Salvatori:2025oib},
although a direct comparison is not straightforward because the two methods are applied in different settings, namely the 
surface-ordered amplitudes in  $\Tr(\phi^3)$ versus the $n$-point correlation functions in scalar quantum field theories here.
In reverse, we expect that many of the concepts uncovered in our purely graph-centred setting, such as the exact solvability of tropicalized scalar quantum field theory, have counterparts within the surfaceology framework, and we hope that the present article helps to open a path for a deeper integration of surfaceology with tropicalized quantum field theory.

Recent empirical studies of perturbation theory at large loop orders~\cite{Balduf:2023ilc,Balduf:2024gvv,Balduf:2024njk,Borinsky:2025ncs} provide additional motivation for algorithms that sample quantum field theoretic quantities directly, without explicitly enumerating Feynman diagrams. In particular, Balduf's analysis \cite[\S 4]{Balduf:2023ilc}  indicates that the distribution of contributions from individual Feynman integrals becomes increasingly heavy-tailed at high loop order. This implies that a naive Monte Carlo strategy---randomly sampling a graph, computing its Feynman integral, and averaging over the results---would suffer from enormous variance and fail to converge efficiently. Such heavy-tailed behaviour makes it unlikely that unbiased sampling over graphs could produce reliable estimates in a reasonable time, showing that more sophisticated methods like the one developed in this work are needed.

The tropical loop equation and the associated non-linear recursion equation for the expansion coefficients 
of the tropicalized $n$-point correlation functions have a relative in the theory of matrix integrals: the loop
equations and the more general Eynard--Orantin topological recursion~\cite{Eynard:2007kz}. It is plausible, but not obvious, that our exact solution of tropicalized massive scalar quantum field theory and the associated non-linear recursion have interpretations as topological recursions.  

From a purely mathematical standpoint, the first part of this article proves a formula for the generating function of certain volumes of moduli spaces of metric graphs \cite{culler1986moduli} (see also \cite{chan2021moduli}). The necessary volume forms are tropicalized versions of parametric Feynman integrands.
The algorithm in the second part samples points from these moduli spaces uniformly with respect to these forms.
Such volume computations have an analogy in the moduli space of Riemann surfaces, thanks to the seminal work of Mirzakhani~\cite{Mirzakhani:2006fta}, which also provided input to the surfaceology framework~(see, e.g.,~\cite[\S 8.2]{arkani2023all}). 
Her recursion computes the volumes of moduli spaces of Riemann surfaces with fixed boundary lengths.  %

In the context of amplitude computations, volume forms over moduli spaces of graphs were previously studied by Berghoff~\cite{Berghoff:2017dyq}. We also remark that the moduli space of graphs naturally emerges as an infrared limit of the moduli space of Riemann surfaces in string perturbation theory \cite{wittenTropical,Tourkine:2013rda} and that tropical geometric methods found applications in string topology \cite{Mandel:2019uae}.

\subsection{Plan of the paper}
Section~\ref{sec:brieftrop} targets readers who are familiar with the path-integral formulation of quantum field theory. There, we briefly summarize the gist of tropicalized quantum field theory. We recall more specific prerequisites on parametric perturbation theory, tropical geometry, and the Hepp bound in Section~\ref{sec:pre}.
In this article, we focus on massive scalar quantum field theories in a Euclidean spacetime. 
However, there is much flexibility concerning these restrictions, and likely the results of this article can be extended in various ways. We give an overview of weaker but sufficient restrictions beyond the massive Euclidean case in Section~\ref{sec:positivity}.
 We discuss tropicalized quantum field theory in depth and prove the tropical loop equation in Section~\ref{sec:tropqft}. For the proof, we define the tropical effective action in Section~\ref{sec:tropeff} as a generating function of Hepp-weighted graphs, which we then prove to be equal to the quantum effective action in the tropical limit. We show that the tropical loop equation can be interpreted as a Dyson--Schwinger equation in Section~\ref{sec:perttropqft} and discuss tropicalized quantum field theory in the critical dimension, where the tropical loop equation degenerates to a non-linear ODE, in Section~\ref{sec:crit}. Remarkably, we find that renormalization and tropicalization are naturally compatible.

    In Section~\ref{sec:algo}, we introduce the sampling algorithm and prove its correctness. 
There, for the benefit of readability, we will narrow our scope from general massive scalar quantum field theories to massive $\phi^k$ theory. 
We will start in Section~\ref{sec:modulimetric} by interpreting perturbation theoretic problems in massive  $\phi^k$ theory as integration problems over moduli spaces of graphs. In Section~\ref{sec:preproc}, we translate the tropical loop equation, the manifestation of the exact solvability of tropicalized quantum field theory, into a non-linear recursion equation for the perturbative coefficients of the $n$-point correlation functions at all orders. This recursion can be solved in polynomial time.
Based on this explicit recursion, we can state the global, polynomial-time sampling algorithm in Section~\ref{sec:poly}, where we also prove its correctness.

A proof-of-concept \texttt{C++} implementation of the algorithm is included in the ancillary files of the arXiv version of this article and available as a GitHub repository.\footnote{\url{https://github.com/michibo/amplitrop}}
In Section~\ref{sec:applications}, we present some illustrative computations we performed with this implementation.
In the final Section~\ref{sec:perspective}, 
we outline possible perspectives and generalizations arising from the open directions highlighted throughout this work, pointing toward future investigations.

\subsection{A brief overview of tropicalized quantum field theory}
\label{sec:brieftrop}
A quantum field theory is typically defined by specifying a space of field  configurations, which consists of sections of some bundle over a spacetime manifold, and an action functional encoding the dynamics of the theory.  A prototypical example is the manifold $\R^D$ with the trivial line bundle, the field configuration space given by functions $\Phi: \R^D \rightarrow \R$, and the action functional,
\begin{align} \label{eq:S} \mathcal S[\Phi,J] = \int_{\R^D} \dd^{D} x \,\bigg( \frac12 \Phi(x) \left(\Box +m^2 \right) \Phi(x) - V[\Phi](x) - J(x) \Phi(x) \bigg), \end{align}
which associates a number to each bosonic field configuration $\Phi,J$ on $\R^D$. We will work with a general scalar potential $V[\Phi](x)=\sum_{k\geq 3} \frac{ \lambda_k'}{k!} \Phi(x)^k$ with coupling constants $\lambda_3', \lambda_4',\ldots$, and the negative Laplacian on $\R^D$: $\Box = - \partial_\mu \partial^\mu$. %

Our opening move is to \emph{deform} each member of the scalar QFT family \eqref{eq:S} via a positive real parameter $\xi$ as follows:
\begin{align} \label{eq:Sxi} \mathcal S^\xi[\Phi,J] = \int_{\R^{D\cdot \xi}} \dd^{D\cdot \xi} x \,\bigg( \frac12 \Phi(x) \left(\Box +m^2 \right)^\xi \Phi(x) - V[\Phi](x) - J(x) \Phi(x) \bigg), \end{align}
where we rescaled the spacetime dimension as $D \mapsto D\cdot \xi$ and powered the propagator as $(\Box + m^2) \mapsto (\Box + m^2)^{\xi}$.
Obviously, $\mathcal S[\Phi]=\mathcal S^1[\Phi]$.
We will use dimensional regularization~\cite{tHooft:1972tcz}. So, the values of the parameters $D$ and $\xi$ will not necessarily be integers. 
It is convenient to pass to dimensionless coupling constants $\lambda_k =\mu^{\xi((k-2)D/2 -k)} \lambda_k'$, where $\mu$ is an auxiliary mass parameter.

The first message of this article is that we can rightfully call $\xi \rightarrow 0^+$ the \emph{tropical limit} of \eqref{eq:S}.

Before diving into the precise meaning of this statement, we briefly recall the path integral workflow. We will not use path integrals in our proofs or later in the article. Here, in this section, they are convenient to provide a quick overview of the essence of tropicalized quantum field theory from a physical perspective. We write
\begin{align*} Z^\xi[J] = \int \exp\left(-\mathcal S^\xi [\Phi,J] \right)\, \mathcal D[\Phi]\, , \end{align*}
for the partition functional that only depends on the source field $J$.
Following a standard quantum field theory stream of logic (see, e.g.,~\cite{Weinberg:1995mt}), we define the 
free energy functional $W^\xi[J] = \log Z^\xi[J]$ and the effective action functional, 
$\Gamma^\xi[\Phi_c]= W^\xi[J[\Phi_c]] - \Phi_c \cdot J[\Phi_c]$, 
 which is the Legendre transform of $W^\xi[J]$ with respect to the classical field $\Phi_c[J] = {\delta W^\xi}/{\delta J}$.
The momentum-space $n$-point correlation functions are then
$ \Gamma^\xi_{n}(p_1,\ldots,p_n)= {\delta^n \Gamma^\xi[\Phi_c]}/({\delta \Phi_c(p_1) \cdots \delta \Phi_c(p_n)}). $

We will show that the $n$-point correlation functions behave well in the tropical limit $\xi \rightarrow 0^+$.
Explicitly, we will prove that for Euclidean momenta $p_1,\ldots,p_n$ and a positive mass $m^2>0$ (see Section~\ref{sec:positivity} for weaker assumptions, which are also sufficient), we find 
\begin{align} \label{eq:deftropeff}  \Gamma^\xi_{n}(p_1,\ldots,p_n) = \mu^{\xi(D+n(1-D/2))} \cdot \delta^{({D\cdot \xi})}\left( \sum_{i=1}^n p_i\right) \cdot \left( \frac{\partial^n \Gamma^\tr}{\partial \varphi^n}\Big|_{\varphi=0} + \smallO(\xi)\right) \text{ as } \xi \rightarrow 0^+ , \end{align}
where $\Gamma^\tr= \Gamma^\tr(D,\varphi,\lambda_3,\lambda_4,\ldots)$ is a dimensionless \emph{function} (and not a functional) in the dimension $D$, the scalar formal variable $\varphi$, and the dimensionless coupling constants $\lambda_3,\lambda_4,\ldots$  We will refer to this function as the \emph{tropical effective action}. %

The key result of this article is that the tropical effective action, $\Gamma^\tr$, is completely determined by the following partial differential equation (PDE),
 which we call the \emph{tropical loop equation}, 
\begin{align} \label{eq:pde} \mathcal P_D \, \Gamma^\tr = \left(1- \frac{\partial^2 \Gamma^\tr}{ \partial \varphi^2} \right)^{-1} - 1, \end{align}
with the linear, first-order differential operator $\mathcal P_D$ on the left-hand side,
\begin{gather} \label{eq:PD} \mathcal P_D = \left( - D - \left(1-\frac{D}{2}\right) \varphi \frac{\partial}{\partial \varphi} +\sum_{k\geq 3} \left( k-D \left(\frac{k}{2}-1\right)\right) \lambda_k \frac{\partial}{\partial \lambda_k} \right), \end{gather}
and the boundary data, which is given by 
\begin{align} \label{eq:boundary} \lim_{\hbar \rightarrow 0}\, \hbar \cdot \Gamma^\tr\left(D, \varphi {\hbar}^{-\frac12}, \hbar^{\frac12} \lambda_3,\hbar^{\frac22} \lambda_4, \hbar^{\frac32} \lambda_5,\ldots\right) = \sum_{k \geq 3} \lambda_k \frac{\varphi^k}{k!}. \end{align}
Eq.~\eqref{eq:boundary} fixes $\Gamma^\tr$ at tree-level
and eq.~\eqref{eq:pde} uniquely determines $\Gamma^\tr$ beyond that (i.e.~at loop-level).
These formulas form the foundation of all other results in this article.

Eq.~\eqref{eq:pde} is a PDE in infinitely many parameters, which often reduces transparently to a PDE with only a finite number of parameters. 
For instance, %
specializing the tropical effective action to the $\phi^k$-theory case, i.e., specifying 
$\Gamma^\tr_{\phi^k} := \Gamma^\tr(D,\varphi, 0, \ldots, 0,\lambda_k, 0,\ldots)$ leads to the PDE
\begin{align} \label{eq:troploopphik} \mathcal P_D^{\phi^k} \, \Gamma^\tr_{\phi^k} = \left(1- \frac{\partial^2 \Gamma^\tr_{\phi^k}}{ \partial \varphi^2} \right)^{-1} - 1, \end{align}
where the differential operator is 
$$
\mathcal P_D^{\phi^k} = 
  \left( - D - \left(1-\frac{D}{2}\right) \varphi
\frac{\partial}{\partial \varphi}
+\left( k-D \left(\frac{k}{2}-1\right)\right) \lambda_k \frac{\partial}{\partial \lambda_k} \right)
$$ and the boundary condition $ \lim_{\hbar \rightarrow 0}\, \hbar \cdot \Gamma^\tr\left(D, \varphi {\hbar}^{-\frac12}, 0,\ldots,\hbar^{\frac{k-2}{2}} \lambda_k, 0,\ldots\right) = \lambda_k \frac{\varphi^k}{k!}.$
Analogously, we get a finite-dimensional PDE problem if we restrict to finitely many non-zero couplings.

We leave the exploration of tropicalized quantum field theories beyond the scalar and massive case to future works. It even seems possible that gauge theories can be  tropicalized. Tropicalized quantum field theory captures only the dominant singular behaviour of the original quantum field theory as a function of the spacetime dimension $D$. Since gauge theories are generally compatible with dimensional regularization, it is plausible that they, too, admit a tropicalization.

\section*{Acknowledgments}
I am deeply grateful to Erik Panzer for countless discussions on the Hepp bound since 2016. I am also very thankful to Paul Balduf, Dario Benedetti, Francis Brown, Freddy Cachazo, Kevin Costello, Gerald Dunne, Davide Gaiotto, Dirk Kreimer, Oliver Schnetz, and Karen Yeats for their valuable comments and insights.
I owe special thanks to Nima Arkani-Hamed, Carolina Figueiredo, and Giulio Salvatori for patient and generous explanations of surfaceology.
I thank the anonymous referee for helpful comments, suggestions, and useful pointers to the literature.

Research at Perimeter Institute is supported in part by the Government of Canada through the Department of Innovation, Science and Economic Development and by the Province of Ontario through the Ministry of Colleges and Universities. 

\section{Preliminaries}
\label{sec:pre}
\subsection{Perturbation theory}
\label{sec:pert}

The (Feynman) graphs typically used in quantum field theory are finite, one-dimensional CW complexes.    Our graphs will have no vertices of~degree $0$ or $2$. Vertices of degree 1 and the adjacent edges have a special role. They are called legs. 
We do not allow edges that only connect a pair of legs.
We reserve the name vertex for vertices that are not legs. %
The edges incident to legs are called external edges, and the term edge is reserved for all other (internal) edges.
A graph is bridgeless if an (internal) edge's removal does not change the number of connected components. It is 1PI (1-particle-irreducible) if it is connected and bridgeless. 
We consider the legs to be labelled uniquely by integers $1,\ldots,n$, and we write $\Aut (G)$ for the automorphism group  of a graph $G$. 
Automorphisms are required to retain the leg labels.
We write 
$\mathcal G_n$ for the set of isomorphism classes of 1PI graphs with $n$ legs.

The $n$-point correlation functions associated with the action~\eqref{eq:S} can be computed perturbatively (i.e.~as a power series expansion) by summing over all 1PI graphs:
\begin{align} \label{eq:defGamma} \Gamma_{n}^\xi(p_1,\ldots,p_n) = \sum_{\substack{G\in \mathcal G_n}} \frac{\prod_{v\in V_G} \left( \mu^{\xi(|v| - {(|v|-2)} D/2)} \lambda_{|v|} \right)} {|\Aut (G)|} I_G^\xi(p_1,\ldots,p_n) , \end{align}
Here, for each of $G$'s vertices $v\in V_G$, we denote its degree by $|v|$. The product over $V_G$ produces the coupling constant monomial associated with the graph $G$. The powers of the auxiliary mass parameter $\mu$ are consequences of the choice of dimensionless coupling constants $\lambda_k$.
We interpret $\Gamma_n^\xi$ as a power series in the formal parameters $\lambda_3,\lambda_4,\ldots$ whose coefficients are functions of $\xi, D, m,p_1,p_2,\ldots$ %

The last term $I_G^\xi(p_1,\ldots,p_n)$ denotes the Feynman integral associated with $G$. It is given by
\begin{align*}  I_G^\xi(p_1,\ldots,p_n)= \prod_{e \in E_G} \left( \int_{\R^{D\cdot\xi}} \frac{\dd^{D\cdot\xi} k_e}{(2\pi)^{D\cdot\xi}} \frac{1}{(k_e^2+m^2)^\xi} \right) \prod_{v \in V_G} \widehat \delta^{({D\cdot\xi})}\left( \sum_{e\in E_G}\mathcal E_{v,e} k_e + \sum_{i=1}^n \mathcal P_{v,i} p_i \right) \end{align*}
The integration over $k_e$ for each of $G$'s edges in $E_G$ extends to the right-most factor,
where $\mathcal E_{v,e}$ is $G$'s incidence matrix computed by fixing an arbitrary orientation of $G$'s edges and setting $\mathcal E_{v,e} = +1/-1$ if the edge $e$ is directed towards/away from the vertex $v$ 
and $\mathcal E_{v,e}=0$ if $e$ is not incident to $v$.
The matrix $\mathcal P_{v,i}$ is computed analogously by setting $\mathcal P_{v,i} =1$ if the $i$-th leg 
connects to vertex $v$ and zero otherwise. 
Here, $\widehat \delta^{({D})}(x) = (2\pi)^D \delta^{(D)}(x)$ is a convenient normalization of the usual $D$-dimensional Dirac $\delta$-function.

Passing to the Feynman parametric representation (see, e.g.,~\cite[Chapters~2--3]{Weinzierl:2022eaz}) and splitting off an overall $\delta$-function yields a compact expression for the $\xi$-deformed Feynman integral 
\begin{align} \begin{gathered} \label{eq:defFeynman} I_G^\xi(p_1,\ldots,p_n) = \mu^{-2\xi\cdot\omega_D(G)} \cdot \widehat \delta^{({D\cdot \xi})}\left( \sum_{i=1}^n p_i\right) \cdot \widetilde I_G^\xi(p_1,\ldots,p_n), \textrm{ where } \\
\widetilde I_G^\xi(p_1,\ldots,p_n) = 1 \textrm{ if } L(G) = 0 \textrm { and } \\
\widetilde I_G^\xi(p_1,\ldots,p_n) =   \frac{\Gamma(\omega_D(G) \cdot \xi) } {\Gamma(\xi)^{|E_G|} (4\pi)^{L(G)\cdot D \cdot \xi /2} } \int_{\mathbb P_{>0}^{E_G}} \left( \frac{{\prod_{e\in E_G} x_e}} { \mathcal U_G(\bb x)^{D /2} \cdot \mathcal V_G(\bb x)^{\omega_D(G)} } \right)^{\xi} \Omega  \text{ if } L(G) \geq 1\, .    \end{gathered} \end{align}
Here, the loop number of $G$ is $L(G)$, which is equal to the rank of $G$'s fundamental group. As $G$ is connected, we have by the usual Euler characteristic formula, $L(G)= |E_G|-|V_G|+1$.
The \emph{superficial degree of divergence} of $G$ is the number $\omega_D(G)=|E_G|-D\cdot L(G)/2$. %
Further, the Euler $\Gamma$-function is defined by $\Gamma(x) = \int_0^\infty t^{x} e^{-t} \frac{\dd t}{t}$,
the domain $\mathbb P_{>0}^{E_G}$ is the 
positive part of $|E_G|-1$ dimensional real projective space $\mathbb P_{>0}^{E_G} = \{ [x_1:\cdots:x_{|E_G|}] \in \mathbb{RP}^{|E_G|-1}: x_e > 0\}$, and 
\begin{align} \label{eq:Omega} \Omega = \sum_{e=1}^{|E_G|} (-1)^{|E_G|-e} {\dd \log x_1} \wedge \cdots \wedge \widehat{ \dd \log x_e} \wedge \cdots \wedge {\dd \log x_{|E_G|}} \end{align}
is the canonical volume form on the projective simplex.
The Symanzik polynomials $\mathcal U_G(\bb x)$, $\mathcal F_G(\bb x)$, and the rational function $\mathcal V_G(\bb x) = \mathcal F_G(\bb x) / \mathcal U_G(\bb x)$ 
respect the combinatorial formulas,
\begin{align} \begin{aligned} \label{eq:UF} \mathcal U_G(\bb x) &= \sum_{T \subset G} \prod_{e \not \in T} x_e &&\text { and }& \mathcal F_G(\bb x) &= \sum_{F \subset G} \frac{p(F)^2}{\mu^2} \prod_{e\not \in F} x_e + \frac{m^2}{\mu^2} \cdot \mathcal U_G(\bb x) \cdot \sum_{e\in E_G} x_e, \end{aligned} \end{align}
where we sum over all spanning trees $T$  and all spanning two-forests $F$ of $G$, respectively, and $p(F)$ is the total momentum flowing between $F$'s two components. In contrast to typical treatments, we chose to extract the mass dimension of the polynomial $\mathcal F_G$ for later convenience.
Note that the expression for $\widetilde I_G^\xi(p_1,\ldots,p_n)$
is dimensionless and explicitly depends on $D$ and $\xi$.

\subsection{Divergences of Feynman integrals and positivity assumptions}
\label{sec:positivity}

The parametric Feynman integral in eq.~\eqref{eq:defFeynman} is divergent for many graphs $G$ and many values of $D$, $m$ and $p_1,p_2, \ldots$ If the external momenta $p_i$ are Euclidean, we have effective convergence criteria (see~\cite{Speer:1975dc}). 
In the Minkowskian case, 
no such combinatorial convergence criteria exist yet. 
Even in the pseudo-Euclidean regime (see, e.g., \cite{Borinsky:2023jdv}), where all the coefficients of $\mathcal F_G$ are positive, only convergence theorems, which require cumbersome computations involving polytopes, can be applied in general  (see, e.g., \cite[Theorem~2.3]{Borinsky:2020rqs}). 
A vanishing mass $m^2$ complicates the question whether a given Feynman integral is convergent or not, and general Minkowski space Feynman integrals can usually only be evaluated unambiguously after a specific deformation of the integration domain via the Feynman $i\varepsilon$ prescription (see, e.g., \cite{Hannesdottir:2022bmo}).

To avoid these additional complications, we make the simplifying assumption that our kinematic parameters are Euclidean and that the mass of the scalar field is positive $m^2 >0$. In fact, the weaker assumptions of the kinematics being in the pseudo-Euclidean regime as defined in~\cite{Borinsky:2023jdv} or the co-positive cone as defined in~\cite{Sturmfels:2025wpg} are sufficient together with the $m^2>0$ assumption or the assumption that the external momenta are sufficiently generic (see, e.g., \cite{Borinsky:2023jdv}). 

Even under all these assumptions, the Feynman integral in eq.~\eqref{eq:defFeynman} is still often divergent.
However, the divergences are well understood. For instance, for a fixed graph $G$, the integral~\eqref{eq:defFeynman} exists if the dimension $D$ is sufficiently small, and it is a meromorphic function in $D$  \cite{Speer:1969sjv}. A meromorphic function has a well-defined analytic continuation to the whole complex plane, so even if~\eqref{eq:defFeynman} is only defined for small $D$, we define it for almost all $D$ via analytic continuation.
\subsection{Tropical algebra}

For a polynomial $p(\bb x) = \sum_{\bb k \in I} a_{\bb k} \prod_{i=1}^n x_i^{k_i}$,
where $I$ is some set of multi-indices $\bb k =(k_1,\ldots,k_n)$ such that $a_{\bb k} \neq 0$ for all $\bb k \in I$,
we define the \emph{tropical approximation} of $p$ as
\begin{align*}  p^\tr(\bb x) = \max_{\bb k\in I} \prod_{i=1}^n x_i^{k_i}. \end{align*}
For example, if $p(\bb x) = x_1^4 + \frac23 x_2^2 x_3^2 + 10^5 x_1 x_3^3$,
then $p^\tr(\bb x) = \max \{ x_1^4, x_2^2 x_3^2, x_1 x_3^3\}$ is the piece-wise
monomial function in $x_1,x_2,x_3,x_4$ that is equal to either $x_1^4, x_2^2 x_3^2,$ or $x_1 x_3^3$ depending on 
which monomial is the largest for the given values of $x_1,x_2,x_3,$ and $x_4$. The tropical approximation disregards the 
precise coefficients of the original polynomial $p(\bb x)$. 
However, the function $p^\tr$ can change drastically if one of the coefficients is set to $0$.
What follows is a key property of $p^\tr$:
\begin{lemma}
\label{lmm:tropapprox}
If $p$ is a homogeneous polynomial in $x_1,\ldots,x_n$ with only positive coefficients, then
there are constants $C_1, C_2 > 0$ such that
\begin{align*} C_1 \leq \frac{p(\bb x)}{p^\tr(\bb x)} \leq C_2 \text{ for all } \bb x \in \mathbb P_{>0}^n, \end{align*}
with $\mathbb P_{>0}^n = \{ [x_1:\ldots:x_n] \in \mathbb {RP}^{n-1} : x_i > 0 \text{ for } i =1,\ldots, n\}$.
\end{lemma}
\begin{proof}
Let 
$C_1 = \min_{\bb k \in I} a_{\bb k}$ and $C_2 = \sum_{\bb k \in I} a_{\bb k}$.
\end{proof}

A stronger version of this statement is proven in \cite{Borinsky:2020rqs}, which also allows for negative 
coefficients of $p$.
The function $p^\tr$ carries the same information as the \emph{Newton polytope} $\mathrm{N}[p]\subset \R^n$ of $p$.
The piece-wise linear function, $\bb y \mapsto \log p^\tr(e^{y_1},\ldots,e^{y_n}) = \max_{\bb v \in \mathrm{N}[p]} \bb v^T \bb y$ is the traditional \emph{tropicalization} of $p$ \cite{MR3287221}.

\subsection{The Hepp bound}
\label{sec:hepp}
We will show that 
we can continuously deform the general scalar quantum field theory~\eqref{eq:S} into a tropicalized quantum field theory, which is perturbatively described by 
tropicalized parametric Feynman integrals (see the right-hand side of eq.~\eqref{eq:tropFI} below). Such tropicalized Feynman integrals were introduced by Panzer in \cite{Panzer:2019yxl} as the \emph{Hepp bound}. 

\begin{proposition}
\label{prop:troplimit}
Provided that all coefficients of the $\mathcal F_G$ polynomial are positive, 
the parametric Feynman integral can be replaced by its tropicalized counterpart
in the limit $\xi\rightarrow 0^+$. That means
\begin{align} \begin{aligned} \label{eq:tropFI} \widetilde I_G^\tr(p_1,\ldots,p_n) :&= \lim_{\xi \rightarrow 0^+} \widetilde I_G^\xi(p_1,\ldots,p_n) \\
&=    \frac{1}{\omega_D(G)} \int_{\mathbb P_{>0}^{E_G}} \frac{{\prod_{e\in E_G} x_e}} { \mathcal U_G^\tr(\bb x)^{D /2} \cdot \mathcal V_G^\tr(\bb x)^{\omega_D(G)} } \, \Omega  \quad \text{ if } L(G) \geq 1\, ,    \end{aligned} \end{align}
and 
$ \widetilde I_G^\tr(p_1,\ldots,p_n) := \lim_{\xi \rightarrow 0^+} \widetilde I_G^\xi(p_1,\ldots,p_n) = I_G^\xi(p_1,\ldots,p_n)= 1 \text{ if } L(G) =0 $,
where  
$\mathcal V_G^\tr(\bb x)= \mathcal F_G^\tr(\bb x) /\mathcal U_G^\tr(\bb x)$ with  the tropical approximations of $\mathcal U_G$ and $\mathcal F_G$ from eq.~\eqref{eq:UF}.
\end{proposition}

\begin{proof}
Let $\xi > 0$ and consider the smooth automorphism $\iota_\xi : \mathbb P_{>0}^{E_G} \rightarrow \mathbb P_{>0}^{E_G}$ given in homogeneous coordinates by $x_k \mapsto x_k^{1/\xi}$.
Pulling back the canonical form $\Omega$ from eq.~\eqref{eq:Omega} along $\iota_\xi$ gives $\iota_\xi^* \Omega = \xi^{1-|E_G|} \Omega$.
Hence, using the transformation $\iota_\xi$ on the
integral expression of \eqref{eq:defFeynman} results in
\begin{align*} \widetilde I_G^\xi(p_1,\ldots,p_n) =  \frac{\xi \cdot \Gamma(\omega_D(G) \cdot \xi)} {(\xi \cdot \Gamma(\xi))^{|E_G|} (4\pi)^{L(G)\cdot D \cdot \xi /2} } \int_{\mathbb P_{>0}^{E_G}} \left( \frac{{\prod_{e\in E_G} x_e^{1/\xi}}} { \mathcal U_G(\bb x^{1/\xi})^{D /2} \cdot \mathcal V_G(\bb x^{1/\xi})^{\omega_D(G)} } \right)^{\xi} \Omega, \end{align*}
where $\mathcal U_G(\bb x^{1/\xi})$ stands for $\mathcal U_G(x_1^{1/\xi},\ldots,x_{|E_G|}^{1/\xi})$ and $\mathcal V_G(\bb x^{1/\xi})$ analogously.

For any polynomial $p(\bb x) = \sum_{\bb k \in I} a_{\bb k} \prod_{i=1}^n x_i^{k_i}$ with only positive coefficients $a_{\bb k} > 0$, we have
\begin{align*} \lim_{\xi\rightarrow 0^+} p(x_1^{1/\xi},\ldots,x_n^{1/\xi})^{\xi} = \lim_{\xi\rightarrow 0^+} \left( \sum_{\bb k \in I} a_{\bb k} \left( \prod_{i=1}^n x_i^{k_i} \right)^{\frac{1}{\xi}} \right)^\xi = \max_{\bb k \in I} \prod_{i=1}^n x_i^{k_i} = p^{\tr}(x_1,\ldots,x_n). \end{align*}
The statement follows after using Lemma~\ref{lmm:tropapprox} and
$\Gamma(\omega_D(G) \cdot \xi) = \frac{1}{\omega_D(G) \cdot \xi} + \bigO(1)$ for $\xi \rightarrow 0$.
\end{proof}

Panzer's Hepp bound beautifully captures the combinatorics of the integral on the right-hand side of eq.~\eqref{eq:tropFI}. The original Hepp bound only addresses the $\omega_D(G) \rightarrow 0$ case in eq.~\eqref{eq:tropFI}. In this case, the \emph{Feynman integral} reduces to a \emph{Feynman period} whose value is explicitly bounded by the Hepp bound. This bound follows from an application of Lemma~\ref{lmm:tropapprox} under the integral sign on the right-hand side of~\eqref{eq:tropFI}. Our generalization of this object to arbitrary values of $\omega_D(G)$ is quite transparent. It gives a bound on the full Feynman integral provided that the values of the external kinematics are controlled. So, we still call it the Hepp bound.

\begin{definition}
\label{def:Hepp}
To each graph $G$, we assign the Hepp bound $\mathcal H_D(G)$, a rational function in $D$ by fixing 
$\mathcal H_D(G) = 1$ if $G$ is a graph without edges (i.e.~a disjoint union of vertices),
and otherwise  
\begin{align} \label{eq:HDedgerec} \mathcal H_D(G) = \frac{1}{\omega_D(G)} \sum_{e\in E_G} \mathcal H_D(G\setminus e), \end{align}
which recursively defines the Hepp bound for all graphs $G$.
(Recall $\omega_D(G) = |E_G| - D \cdot L(G)/2$.)
\end{definition}
Our definition differs slightly from the original one in \cite{Panzer:2019yxl},
which defined the Hepp bound via a sum over all permutations  and included an overall factor of $\omega_D(G)$. 
See Proposition~6.8 \cite{Borinsky:2020rqs} for a proof that Panzer's original definition fulfills eq.~\eqref{eq:HDedgerec}.

Our viewpoint on the Hepp bound is similar to Salvatori's in \cite[Section~3.1]{Salvatori:2025oib}, who put forward an elegant perspective on the Hepp bound that focuses on the identity~\eqref{eq:HDedgerec}. For the reader's convenience, we will reprove and recombine results on the Hepp bound in this section, both because we will need them later  and to keep our exposition self-contained.
\begin{figure}
\begin{align*} \mathcal H_D\left( \begin{tikzpicture}[every path/.style={thick},baseline={(0,-.1)}] \def\rad{.6} \def\rud{.6} \def\radl{0} \coordinate (v1) at (-\rud,0); \coordinate (v2) at (\rud,0); \coordinate (v01) at (-\radl,0); \coordinate (v02) at (\radl,0); \filldraw[pattern color=col3,pattern=north east lines] (v1) circle(\rad); \filldraw[pattern color=col3,pattern=north east lines] (v2) circle(\rad); \fill (v01) circle(2pt); \fill (v02) circle(2pt); \node[fill=white,circle] at (v1) {$G_1$}; \node[fill=white,circle] at (v2) {$G_2$}; \end{tikzpicture} \right) = \mathcal H_D\left( \begin{tikzpicture}[every path/.style={thick},baseline={(0,-.1)}] \def\rad{.6} \def\rud{.8} \def\radl{.2} \coordinate (v1) at (-\rud,0); \coordinate (v2) at (\rud,0); \coordinate (v01) at (-\radl,0); \coordinate (v02) at (\radl,0); \filldraw[pattern color=col3,pattern=north east lines] (v1) circle(\rad); \filldraw[pattern color=col3,pattern=north east lines] (v2) circle(\rad); \fill (v01) circle(2pt); \fill (v02) circle(2pt); \node[fill=white,circle] at (v1) {$G_1$}; \node[fill=white,circle] at (v2) {$G_2$}; \end{tikzpicture} \right) = \mathcal H_D\left( \begin{tikzpicture}[every path/.style={thick},baseline={(0,-.1)}] \def\rad{.6} \def\rud{.8} \def\radl{.2} \coordinate (v1) at (-\rud,0); \coordinate (v01) at (-\radl,0); \filldraw[pattern color=col3,pattern=north east lines] (v1) circle(\rad); \fill (v01) circle(2pt); \node[fill=white,circle] at (v1) {$G_1$}; \end{tikzpicture} \right) \cdot \mathcal H_D\left( \begin{tikzpicture}[every path/.style={thick},baseline={(0,-.1)}] \def\rad{.6} \def\rud{.8} \def\radl{.2} \coordinate (v2) at (\rud,0); \coordinate (v02) at (\radl,0); \filldraw[pattern color=col3,pattern=north east lines] (v2) circle(\rad); \fill (v02) circle(2pt); \node[fill=white,circle] at (v2) {$G_2$}; \end{tikzpicture} \right) \end{align*}
\caption{Factorization property of the Hepp bound}
\label{fig:factor}
\end{figure}
\begin{proposition}
\label{prop:factor}
The Hepp bound factors for disconnected graphs and graphs with a cut vertex. %
Let $G_1$ and $G_2$ be graphs, $G_1\sqcup G_2$ their disjoint union, and $H$ the result of gluing together $G_1$ and $G_2$ at one vertex. 
The Hepp bound fulfills, as illustrated in Figure~\ref{fig:factor},
\begin{align*}  \mathcal H_D(H) = \mathcal H_D(G_1 \sqcup G_2) = \mathcal H_D(G_1) \cdot \mathcal H_D(G_2). \end{align*}
\end{proposition}
This follows from the more general Proposition~2.24 \cite{{Panzer:2019yxl}} (see also~\cite[Eq.~(3.16)]{Salvatori:2025oib}). We replicate the proof here, adapted to the present notation and specific case.
\begin{proof}
The statement is trivial if $H$, $G_1$ and $G_2$ have no edges, as the Hepp bound is always $1$ then. We proceed by induction on the number of edges and assume that the statement holds for all graphs $G_1$ and $G_2$ such that $G_1\sqcup G_2$ has $n-1$ edges. For pairs $G_1,G_2$, where $G_1\sqcup G_2$ has $n$ edges, we get from eq.~\eqref{eq:HDedgerec} that 
$$
\mathcal H_D(G_1\sqcup G_2) = 
\frac{1}{\omega_D(G_1 \sqcup G_2)}
\left(
\sum_{e_1\in E_{G_1}} 
\mathcal H_D(G_1 \setminus e_1 \sqcup G_2 )
+
\sum_{e_2\in E_{G_2}} 
\mathcal H_D(G_1 \sqcup G_2 \setminus e_2 )
\right).
$$
On the right-hand side, we may use the induction hypothesis. So,
\begin{align*} \mathcal H_D(G_1\sqcup G_2) &= \frac{1}{\omega_D(G_1 \sqcup G_2)} \left( \sum_{e_1\in E_{G_1}} \mathcal H_D(G_1 \setminus e_1) \mathcal H_D( G_2 ) + \sum_{e_2\in E_{G_2}} \mathcal H_D(G_1) \mathcal H_D( G_2 \setminus e_2 ) \right) \\
&= \frac{\omega_D(G_1) + \omega_D(G_2)}{\omega_D(G_1 \sqcup G_2)} \mathcal H_D(G_1) \mathcal H_D( G_2 ), \end{align*}
where we used eq.~\eqref{eq:HDedgerec} again in the second step.
The statement follows from the definition of $\omega_D(G)$.
The argument works analogously for graphs $H$ as given in the statement.  
\end{proof}

\begin{proposition}
\label{prop:hepp}
For a 1PI graph $G$, any set of external momenta $p_1,\ldots, p_n$ and $m^2>0$ such that all coefficients of $\mathcal F_G$ are non-negative, we have 
\begin{gather} \label{eq:Ddecomposition} \widetilde I_G^\tr(p_1,\ldots,p_n) = \HE_D(G). \end{gather}
\end{proposition}
This statement also follows from the argument after Lemma~2.8 in 
\cite{{Panzer:2019yxl}}. 
We prove it in detail using the recursive Definition~\ref{def:Hepp},
as the argument will be important later for the discussions in Section~\ref{sec:algo}.
It might be helpful for the reader to come back to this section while studying the proofs in Section~\ref{sec:poly}.

As Salvatori in \cite[\S 3.1]{Salvatori:2025oib}, we pass from projective to cubical coordinates in the integral representation of $\widetilde I_G^\tr$ from Proposition~\ref{prop:troplimit}.
Let $C_{e}(E_G) = \{ \bb \xiz \in [0,1]^{E_G}: \xiz_{e'} < \xiz_{e} \text{ for all } e' \in E_G\setminus e\}$
be the subset of the $|E_G|$-cube,
where the $e$-th coordinate is the largest.

\begin{lemma}
\label{lmm:kruskal}
Let $G$ be a 1PI graph 
and $e \in E_G$.
If $\bb \xiz \in C_e(E_G)$,
then $\mathcal U^\tr_G(\bb \xiz) = \xiz_{e} \cdot \mathcal U^\tr_{G\setminus e}(\bb \xiz)$.
\end{lemma}
\begin{proof}
For a given set of cubical edge lengths 
$\bb \xiz \in C_e(E_G)$, 
where $z_e$ is the largest coordinate, 
we call a spanning tree $T_{\mathrm{max}}$ of $G$ \emph{maximal}
if it fulfills 
$\prod_{e'\not\in T} \xiz_{e'} \leq \prod_{e'\not\in T_{\mathrm{max}}} \xiz_{e'}$ for all spanning trees $T \subset G$.
Assume that a maximal spanning tree $T_{\mathrm{max}}$ contains the edge $e$. Because $G$ is 1PI, $G\setminus e$ is a connected graph.
So,  there must be an edge $h\in E_G\setminus e$ 
such that $(T_{\mathrm{max}} \setminus e) \cup h$ is a spanning tree. This spanning tree fulfills 
$\prod_{e'\not\in (T_{\mathrm{max}} \setminus e) \cup h} \xiz_{e'} = \frac{\xiz_e}{\xiz_h} \prod_{e'\not\in T_{\mathrm{max}}} \xiz_{e'} > \prod_{e'\not\in T_{\mathrm{max}}} \xiz_{e'}$
resulting in a contradiction. So, there is no maximal spanning tree that does not contain $e$. 
By definition $\mathcal U^\tr_G(\bb \xiz) = \max_{T\subset G} \prod_{e'\not\in T} \xiz_{e'} = \prod_{e'\not \in T_{\mathrm{max}}} \xiz_{e'} = \xiz_{e} \max_{T\subset G\setminus e} \prod_{e'\not\in T} \xiz_{e'} $.
\end{proof}

\begin{corollary}
\label{cor:HDcube}
For all graphs $G$ with $E_G \neq \emptyset$, we have
\begin{align*} \mathcal H_D(G) = \int_{[0,1]^{E_G}} \frac{\dd \xiz_1\cdots \dd \xiz_{|E_G|}} {\mathcal U_G^\tr(\bb \xiz)^{D/2}}\, . \end{align*}
\end{corollary}
Up to a prefactor, this is Salvatori's formula~\cite[(3.12)]{Salvatori:2025oib}.
\begin{proof}
We can apply the decomposition of the cube $[0,1]^{E_G}$ into subsets $\bigsqcup_e C_e(E_G)$ (up to sets of measure zero) to the integral on the right-hand side of the stated equation to get
\begin{gather} \notag \int_{[0,1]^{E_G}} \frac{\dd \xiz_1\cdots \dd \xiz_{|E_G|}} {\mathcal U_G^\tr(\bb \xiz)^{D/2}} = \sum_{e\in E_G} \int_{C_e(E_G)} \frac{\dd \xiz_1\cdots \dd \xiz_{|E_G|}} {\mathcal U_G^\tr(\bb \xiz)^{D/2}}. \intertext{Applying Lemma~\ref{lmm:kruskal} under the integral sign and passing to cubical edge coordinates on $G\setminus e$ given by $\xiz_h' = \xiz_e \xiz_h$ for all $h \in E_{G}\setminus e$ on the right-hand side results in } \label{eq:eUG} \int_{[0,1]^{E_G}} \frac{\dd \xiz_1\cdots \dd \xiz_{|E_G|}} {\mathcal U_G^\tr(\bb \xiz)^{D/2}} = \sum_{e\in E_G} \int_0^1 \frac{\dd \xiz_e}{\xiz_e} \xiz_e^{|E_G|-L(G)D/2} \int_{ [0,1]^{E_{G\setminus e}} } \frac{\dd \xiz_1'\cdots \widehat {\dd \xiz_e'} \cdots \dd \xiz_{|E_G|}'} {\mathcal U_{G\setminus e}^\tr(\bb \xiz')^{D/2}}. \end{gather}
We recover the recursive identity \eqref{eq:HDedgerec},
because $|E_G| - L(G)D/2= \omega_D(G)$.
It can easily be checked that the statements hold for graphs $G$ where $|E_G|=1$. So, the statement follows.
\end{proof}

\begin{lemma}
\label{lmm:Vmax}
For a 1PI graph $G$, any set of external momenta $p_1,\ldots, p_n$ and $m^2>0$ such that all coefficients of $\mathcal F_G$ are non-negative,
we have
$\mathcal V_G^\tr(\bb x) = \max_{e\in E_G} x_e$.
\end{lemma}
This statement also follows from the more general considerations of \cite[Theorem 1.1]{Tellander:2021xdz}.
\begin{proof}
For a graph $G$ and one of its two-forests $F$ contributing to the sum in the formula for $\mathcal F_G$ in~\eqref{eq:UF}, we can find an edge $e_F$ such that $F \cup e_F$ is a spanning tree $T_F$ of $G$. Hence, 
$\prod_{e \notin F} x_e = x_{e_F} \prod_{e \notin F \cup e_F} x_e \leq x_{e_F} \cdot \mathcal U_G^\tr(\bb x)$, where we used that $\mathcal U_G^\tr(\bb x) = \max_{T \subset G} \prod_{e \notin T} x_e$ with the maximum over all spanning trees of $G$. As $m^2 > 0$, all monomials of the form $x_{e'} \prod_{e\notin T} x_e$ with some spanning tree $T$ and edge $e'$ of $G$ appear in the polynomial $\mathcal F_G$. It follows that 
$\mathcal F_G^\tr(\bb x) = (\max_{e\in E_G} x_e) \mathcal \cdot \mathcal U_G^\tr(\bb x)$ and therefore $\mathcal V_G^\tr(\bb x) = \max_{e\in E_G} x_e$. 
\end{proof}

\begin{proof}[Proof of Proposition~\ref{prop:hepp}]
Recall the formula for $\widetilde I_G^\tr$ from Proposition~\ref{prop:troplimit}.
So, by Lemma~\ref{lmm:Vmax},
$$
\widetilde 
I_G^\tr(p_1,\ldots,p_n) =
\frac{1}{\omega_D(G)}
\int_{\mathbb P_{>0}^{E_G}}
\frac{{\prod_{e\in E_G} x_e}}
{
\mathcal U_G^\tr(\bb x)^{D /2} \cdot
(\max_{e\in E_G} x_e)^{\omega_D(G)} 
}
\,
\Omega.
$$
Let $P_e = \{ \bb x \in \mathbb P_{>0}^{E_G}: x_{e'} < x_e \text{ for all } e'\in E_G\setminus e\}$.
We have
$$
\widetilde 
I_G^\tr(p_1,\ldots,p_n) =
\frac{1}{\omega_D(G)}
\sum_{e'\in E_G}
\int_{P_{e'}}
\frac{{\prod_{e\in E_G} x_e}}
{
\mathcal U_G^\tr(\bb x)^{D /2} \cdot
(x_{e'})^{\omega_D(G)} 
}
\,
\Omega.
$$
There is a diffeomorphism
$(0,1)^{E_{G\setminus e'}} \rightarrow P_{e'}$
given by $x_e = \xiz_e$ for $e \neq e'$ and $x_{e'}=1$.
Transforming the integral above into cubical form 
using this diffeomorphism 
results in the right-hand side of 
eq.~\eqref{eq:eUG}, which equals $\mathcal H_D(G)$.
\end{proof}

\section{Tropicalized quantum field theory}
\label{sec:tropqft}
\subsection{The tropical effective action}
\label{sec:tropeff}

With all these tools at hand, we are now ready to give an explicit definition of the tropicalized effective action as a generating function of the Hepp bound 
summed over all graphs. We will then prove that this definition reproduces the 
 quantum effective action in the tropical limit.

Let $\mathcal G$ be the set of all isomorphism classes of 1PI graphs with an arbitrary number of legs, where the legs are now unlabelled. We write $\Aut'(G)$ for the group of automorphisms of $G$ that are allowed to permute the legs and  $n(G)$ for the number of legs of a graph $G \in \mathcal G$.
The tropical effective action $\Gamma^\tr$ is the generating function of the set $\mathcal G$ weighted by the Hepp bound,
\begin{align} \begin{gathered} \label{eq:gammatr} \Gamma^\tr(D,\varphi,\lambda_3,\lambda_4,\ldots) = \sum_{G \in \mathcal G} \frac{\varphi^{n(G)} \prod_{v\in V_G} \lambda_{|v|}} {|\Aut' (G)|}\, \HE_D(G) \in \Q(D)[\![\varphi,\lambda_3,\lambda_4,\ldots]\!]\,. \end{gathered} \end{align}
Under the positivity assumptions outlined in Section~\ref{sec:positivity} (pseudo-Euclidean regime and $m^2>0$),
we can combine the equations~\eqref{eq:defGamma}, \eqref{eq:defFeynman}, \eqref{eq:tropFI}, and \eqref{eq:Ddecomposition} into a formula for the effective action in the tropical limit as given in eq.~\eqref{eq:deftropeff}, which we officially record in the following:
\begin{theorem}
\label{thm:trop}
Assuming Euclidean momenta $p_1,\ldots,p_n$ and a mass $m^2 >0$,  the effective action of the general scalar quantum field theory defined in~eq.~\eqref{eq:Sxi} fulfills the tropical deformation limit,
\begin{align*}  \Gamma^\xi_{n}(p_1,\ldots,p_n) = \mu^{\xi(D+n(1-D/2))} \cdot \delta^{({D\cdot \xi})}\left( \sum_{i=1}^n p_i\right) \cdot \left( \frac{\partial^n \Gamma^\tr}{\partial \varphi^n}\Big|_{\varphi=0} + \smallO(\xi)\right) \text{ as } \xi \rightarrow 0^+ . \end{align*}
\end{theorem}

\begin{proof}
A standard argument using the orbit stabilizer theorem implies that 
\begin{align*} \frac{\partial^n \Gamma^\tr}{\partial \varphi^n} \big|_{\varphi=0} = n! \sum_{\substack{G \in \mathcal G\\n(G) = n}} \frac{\prod_{v\in V_G} \lambda_{|v|}} {|\Aut'(G)|}\, \mathcal H_D(G) = \sum_{G \in \mathcal G_n} \frac{\prod_{v\in V_G} \lambda_{|v|}} {|\Aut(G)|}\, \mathcal H_D(G), \end{align*}
where we fix the leg labels in $\mathcal G_n$ and $\Aut(G)$.
The statement follows after using the formulas
\eqref{eq:defGamma}, \eqref{eq:defFeynman}, \eqref{eq:tropFI}, \eqref{eq:Ddecomposition}
 and confirming the power of $\mu$.
\end{proof}

\subsection{The tropical loop equation}
\label{eq:looprecursion}

In this section, we will prove the tropical loop equation~\eqref{eq:pde},
which (up to a tree-level boundary term) completely fixes $\Gamma^\tr$ and thereby the quantum effective action in the tropical limit. We start by giving  a combinatorial interpretation of \eqref{eq:pde}'s right-hand side.

Recall that a bridge of a graph is an edge whose removal increases the number of connected components,
and that a connected graph without bridges is a 1PI graph.
\begin{definition}
\label{def:beaded}
A connected graph $G$, with two or more legs, exactly two of which have a special colour, is a \emph{beaded graph} if it is either a 1PI graph or if cutting any of its bridges and giving the special colour to both newly created legs from the cut yields a pair of beaded graphs.  
\end{definition}
Beaded graphs can be formed by stringing 1PI two-point graphs together like beads on a string. Note, however, that we allow, in a general beaded graph, each 1PI component to have an arbitrary number of non-special legs.
Figure~\ref{fig:beaded} shows examples of beaded graphs. The special legs are drawn as doubled lines. Note that we can string 
any number of vertices together to form such a graph. The last graph 
in the figure is not beaded, because it has a bridge
whose removal does not separate the two special legs.
\begin{figure}
\begin{align*} \def\rad{.3} \def\rud{.5} \begin{tikzpicture}[every path/.style={thick},baseline={(0,-.1)}] \coordinate (v1i) at (\rad,0); \coordinate (v1m) at (0,0); \coordinate (v1o) at (-\rad,0); \coordinate (v1u) at (0,\rad); \coordinate (v1d) at (0,-\rad); \draw[double] (v1i) -- (v1m); \draw[double] (v1o) -- (v1m); \draw (v1u) -- (v1m); \draw (v1d) -- (v1m); \fill (v1m) circle(2pt); \end{tikzpicture} \, , \begin{tikzpicture}[every path/.style={thick},baseline={(0,-.1)}] \coordinate (v1m) at (0,0); \coordinate (v2m) at (\rud,0); \coordinate (v1o) at (-\rad,0); \coordinate (v1i) at ({\rud+\rad},0); \coordinate (v1u) at (0,\rad); \coordinate (v1d) at (0,-\rad); \coordinate (v2u) at (\rud,\rad); \draw[double] (v1i) -- (v2m); \draw[double] (v1o) -- (v1m); \draw (v1m) -- (v2m); \draw (v1u) -- (v1m); \draw (v1d) -- (v1m); \fill (v1m) circle(2pt); \fill (v2m) circle(2pt); \draw (v2u) -- (v2m); \end{tikzpicture} \, , \begin{tikzpicture}[every path/.style={thick},baseline={(0,-.1)}] \coordinate (v1ii) at ({-\rud-\rad},0); \coordinate (v1i) at (-\rud,0); \coordinate (v1m) at (0,0); \coordinate (v2o) at (\rud,0); \coordinate (v2oo) at ({\rud+\rad},0); \draw (v1m) circle(\rud); \draw[double] (v1ii) -- (v1i); \draw[double] (v2oo) -- (v2o); \fill (v1i) circle(2pt); \fill (v2o) circle(2pt); \end{tikzpicture} \, , \begin{tikzpicture}[every path/.style={thick},baseline={(0,-.1)}] \coordinate (Lcenter) at (-{1.5*\rud},0); \coordinate (Rcenter) at ({1.5*\rud},0); \coordinate (Lright) at ({-\rud/2},0); \coordinate (Rleft) at ({\rud/2},0); \coordinate (Louter) at ({-2.5*\rud-\rad},0); \coordinate (Router) at ({2.5*\rud+\rad},0); \coordinate (v1) at ($(Rcenter) + (60:\rud)$); \coordinate (v2) at ($(Rcenter) + (120:\rud)$); \coordinate (v3) at ($(Rcenter) + (-60:\rud)$); \coordinate (v4) at ($(Rcenter) + (-120:\rud)$); \coordinate (v4i) at ($(Rcenter) + (-120:{\rud+\rad})$); \coordinate (v5) at ($(Lcenter) + (90:\rud)$); \coordinate (v5i) at ($(Lcenter) + (90:{\rud+\rad})$); \draw (v1) -- (v4); \draw[white,line width=5pt] (v3) -- (v2); \draw (v3) -- (v2); \draw (Lcenter) circle(\rud); \draw (Rcenter) circle(\rud); \draw (Lright) -- (Rleft); \draw[double] (Louter) -- ({-2.5*\rud},0); \draw[double] (Router) -- ({2.5*\rud},0); \draw (v5) -- (v5i); \draw (v4) -- (v4i); \fill (Lright) circle(2pt); \fill (Rleft) circle(2pt); \fill ({2.5*\rud},0) circle(2pt); \fill ({-2.5*\rud},0) circle(2pt); \fill (v1) circle(2pt); \fill (v2) circle(2pt); \fill (v3) circle(2pt); \fill (v4) circle(2pt); \fill (v5) circle(2pt); \end{tikzpicture} \, , \begin{tikzpicture}[every path/.style={thick},baseline={(0,-.1)}] \coordinate (Lcenter) at (-{1.5*\rud},0); \coordinate (Rcenter) at ({2.5*\rud},0); \coordinate (Lright) at ({-\rud/2},0); \coordinate (Rleft) at ({1.5*\rud},0); \coordinate (Louter) at ({-2.5*\rud-\rad},0); \coordinate (Router) at ({4.5*\rud+\rad},0); \coordinate (v1) at (\rud/2,0); \coordinate (v1u) at ($(v1) + (90:\rad)$); \coordinate (v1i) at ($(v1) + (-60:\rad)$); \coordinate (v1o) at ($(v1) + (-120:\rad)$); \draw (v1) -- (v1u); \draw (v1) -- (v1i); \draw (v1) -- (v1o); \coordinate (v5) at ($(Lcenter) + (90:\rud)$); \coordinate (v5i) at ($(Lcenter) + (90:{\rud+\rad})$); \draw (Lcenter) circle(\rud); \draw (Rcenter) circle(\rud); \draw (Lright) -- (Rleft); \draw[double] (Louter) -- ({-2.5*\rud},0); \draw[double] (Router) -- ({4.5*\rud},0); \draw ({4.5*\rud},0) -- ({3.5*\rud},0); \draw ({4.5*\rud},0) -- ({4.5*\rud},\rad); \draw ({4.5*\rud},0) -- ({4.5*\rud},-\rad); \draw (v5) -- (v5i); \fill (Lright) circle(2pt); \fill (Rleft) circle(2pt); \fill ({4.5*\rud},0) circle(2pt); \fill ({3.5*\rud},0) circle(2pt); \fill ({-2.5*\rud},0) circle(2pt); \fill (v5) circle(2pt); \fill (v1) circle(2pt); \end{tikzpicture} \, ; \quad \quad \begin{tikzpicture}[every path/.style={thick},baseline={(0,-.1)}] \coordinate (v1m) at (0,0); \coordinate (v11) at (-\rud,0); \coordinate (v11i) at ({-\rud-\rad},0); \coordinate (v12) at (+\rud,0); \coordinate (v12i) at ({+\rud+\rad},0); \coordinate (v13) at (0,-\rud); \coordinate (v13i) at (0,{-\rud-\rad}); \coordinate (v14) at (0,+\rud); \coordinate (v2m) at (0,{3*\rud}); \coordinate (v21) at (0,{4*\rud}); \coordinate (v21i) at (0,{4*\rud+\rad}); \coordinate (v22) at (-\rud,{3*\rud}); \coordinate (v23) at (0,{2*\rud}); \coordinate (v24) at (+\rud,{3*\rud}); \draw[double] (v11i) -- (v11); \draw[double] (v12i) -- (v12); \fill (v11) circle(2pt); \fill (v12) circle(2pt); \fill (v13) circle(2pt); \fill (v14) circle(2pt); \fill (v21) circle(2pt); \fill (v22) circle(2pt); \fill (v23) circle(2pt); \fill (v24) circle(2pt); \draw (v13) -- (v13i); \draw (v1m) circle(\rud); \draw (v2m) circle(\rud); \draw (v14) -- (v23); \draw (v22) -- (v24); \draw (v21) -- (v21i); \end{tikzpicture} \end{align*}
\caption{
Five examples of beaded graphs and 
one example of a non-beaded graph.
}
\label{fig:beaded}
\end{figure}
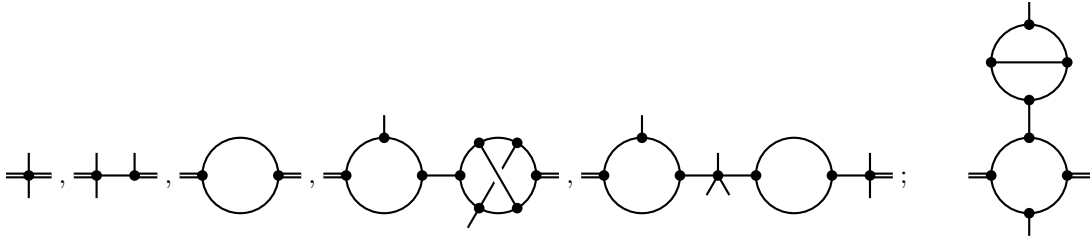

Let $\mathcal{B}$ be the set of isomorphism classes of beaded graphs with two specially coloured legs and an arbitrary number of uncoloured legs. For a beaded graph $G$, let $\Aut_{\mathcal B}'(G)$ be the group of automorphisms of $G$ which do not mix special and non-special legs, but are allowed to permute legs otherwise. Basic generatingfunctionology proves the following formula:
\begin{proposition}
\label{prop:beaded}
The generating function of beaded graphs weighted by their Hepp bound is
\begin{align*} \sum_{ G\in \mathcal B } \frac{\varphi^{n(G)-2} \prod_{v\in V_G} \lambda_{|v|}} {|\Aut_{\mathcal B}' (G)|} \,\mathcal H_D(G) = \frac12 \left( \left( 1 - \frac{\partial^2 \Gamma^\tr}{\partial \varphi^2} \right)^{-1} - 1 \right). \end{align*}
\end{proposition}
\begin{proof}
Almost all involved combinatorial objects are \emph{labelled structures} from the perspective of enumerative combinatorics and generatingfunctionology (see, e.g., \cite[A.II]{FlajoletSedgewick:2009}). It is also well-known in quantum field theory that taking derivatives of a generating function passes to the corresponding pointed object (for a combinatorial viewpoint, see \cite[A.II.6.1]{FlajoletSedgewick:2009}). For instance, 
${\partial^2 \Gamma^\tr}/{\partial \varphi^2}$ is the generating function of 
Hepp bound weighted graphs with two or more legs, two of which are fixed and labelled (i.e.~pointed).

Recall the factorization property of the Hepp bound from Proposition~\ref{prop:factor}.
It follows from Definition~\ref{def:Hepp}
that for a graph with only a single edge, e.g.,
$ \begin{tikzpicture}[baseline={(0,-.1)}] \coordinate (v1) at (-.2,0); \coordinate (v2) at (.2,0); \draw (v1) -- (v2); \fill (v1) circle(1.5pt); \fill (v2) circle(1.5pt); \end{tikzpicture} $,
we have 
$\mathcal H_D( \begin{tikzpicture}[baseline={(0,-.1)}] \coordinate (v1) at (-.2,0); \coordinate (v2) at (.2,0); \draw (v1) -- (v2); \fill (v1) circle(1.5pt); \fill (v2) circle(1.5pt); \end{tikzpicture} ) = 1$,
because 
$\omega_D( \begin{tikzpicture}[baseline={(0,-.1)}] \coordinate (v1) at (-.2,0); \coordinate (v2) at (.2,0); \draw (v1) -- (v2); \fill (v1) circle(1.5pt); \fill (v2) circle(1.5pt); \end{tikzpicture} ) = 1 - 0 = 1$.
Therefore, for a beaded graph $H$ 
that is formed by stringing 
the 1PI graphs $G_1,G_2,\ldots$ together,
we have $\mathcal H_D(H) = \mathcal H_D(G_1) \cdot \mathcal H_D(G_2) \cdots$

We can string any number of two-pointed 1PI graphs together along their distinct legs to form a beaded graph.
By the factorization property, the generating function ${\partial^2 \Gamma^\tr}/{\partial \varphi^2} + \left({\partial^2 \Gamma^\tr}/{\partial \varphi^2}\right)^2 + \left({\partial^2 \Gamma^\tr}/{\partial \varphi^2}\right)^3 +\ldots$ 
of such strings overcounts beaded graphs by a factor of two, because the two endpoints
of the string are not interchangeable. Accounting for the additional 2-fold symmetry and summing the geometric series gives the stated formula.
\end{proof}

This proposition gives a combinatorial interpretation for the right-hand side of~\eqref{eq:pde}. We will continue with the left-hand side, which involves the linear differential operator $\mathcal P_D$. 
 The following lemma explains the somewhat ad hoc-seeming shape of this operator.

\begin{lemma}
\label{lmm:PDdiag}
For any connected graph $G$, we have
\begin{align*}  \mathcal P_D \left( \varphi^{n(G)} \prod_{v\in V_G} \lambda_{|v|} \right) = 2 \omega_D(G) \cdot \varphi^{n(G)} \prod_{v\in V_G} \lambda_{|v|}, \end{align*}
where  $\mathcal P_D$ is defined in eq.~\eqref{eq:PD}
and we recall that $n(G)$ is the number of legs of $G$.
\end{lemma}
\begin{proof}
By the Euler characteristic formula and the total
number of half-edges of a connected graph with $n(G)$
legs, we have 
$L(G) = |E_G| - |V_G| + 1$,  and 
$\sum_{v \in V_G} |v| = n(G) + 2 |E_G|$. Hence,
\begin{gather*}  \omega_D(G) = |E_G| - L(G) \cdot D/2 =     - \frac{D}{2} - \left(1-\frac{D}{2}\right) \frac{n(G)}{2} + \sum_{v \in V_G} \left( \frac{|v|}{2}-\frac{D}{2} \left(\frac{|v|}{2}-1\right)\right). \end{gather*}
Now, replace $n(G)$ with $\varphi \frac{\partial}{\partial \varphi}$,
$\sum_{v\in V_G}1$ with $\sum_{k \geq 3} \lambda_k \frac{\partial}{\partial \lambda_k}$,
and $\sum_{v\in V_G}|v|$
with 
$\sum_{k \geq 3} k \lambda_k \frac{\partial}{\partial \lambda_k}$.
\end{proof}

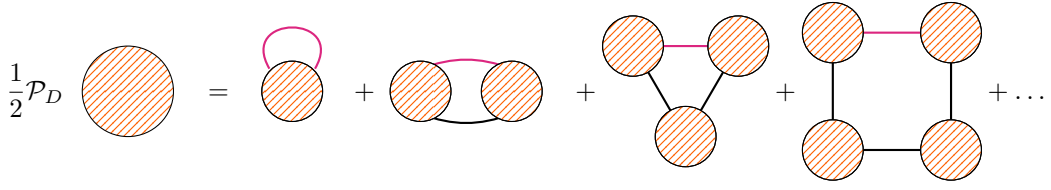
\begin{figure}
\begin{gather*} \frac12 \mathcal P_D ~~ \begin{tikzpicture}[baseline={(0,-.1)}] \filldraw[pattern color=col3,pattern=north east lines] (0,0) circle (.6cm); \end{tikzpicture} \quad =   \begin{tikzpicture}[baseline={(0,-.1)}] \draw[thick, col4, looseness=3.6, out=60, in=120] ($(0,.3)+(.3,0)$) to ($(0,.3)+(-.3,0)$); \draw[fill=white] (0,0) circle (.4cm); \filldraw[pattern color=col3,pattern=north east lines] (0,0) circle (.4cm); \end{tikzpicture} +~  \begin{tikzpicture}[baseline={(0,-.1)}] \coordinate (v1) at (-.6,0); \coordinate (v2) at (.6,0); \draw[thick, col4, bend left=20] (-.6,.3) to (.6,.3); \draw[thick, bend right=20] (-.6,-.3) to (.6,-.3); \draw[fill=white] (v1) circle (.4cm); \filldraw[pattern color=col3,pattern=north east lines] (v1) circle (.4cm); \draw[fill=white] (v2) circle (.4cm); \filldraw[pattern color=col3,pattern=north east lines] (v2) circle (.4cm); \end{tikzpicture} \quad +   \begin{tikzpicture}[baseline={(0,-.3)}] \coordinate (v1) at (30:0.8); \coordinate (v2) at (150:0.8); \coordinate (v3) at (270:0.8); \draw[thick,col4] (v1) -- (v2); \draw[thick] (v2) -- (v3); \draw[thick] (v3) -- (v1); \draw[fill=white] (v1) circle (.4cm); \filldraw[pattern=north east lines, pattern color=col3] (v1) circle (.4cm); \draw[fill=white] (v2) circle (.4cm); \filldraw[pattern=north east lines, pattern color=col3] (v2) circle (.4cm); \draw[fill=white] (v3) circle (.4cm); \filldraw[pattern=north east lines, pattern color=col3] (v3) circle (.4cm); \end{tikzpicture} +   \begin{tikzpicture}[baseline={(0,-.1)}] \coordinate (v1) at (45:1.1); \coordinate (v2) at (135:1.1); \coordinate (v3) at (225:1.1); \coordinate (v4) at (-45:1.1); \draw[thick,col4] (v1) -- (v2); \draw[thick] (v2) -- (v3); \draw[thick] (v3) -- (v4); \draw[thick] (v4) -- (v1); \draw[fill=white] (v1) circle (.4cm); \filldraw[pattern color=col3,pattern=north east lines] (v1) circle (.4cm); \draw[fill=white] (v2) circle (.4cm); \filldraw[pattern color=col3,pattern=north east lines] (v2) circle (.4cm); \draw[fill=white] (v3) circle (.4cm); \filldraw[pattern color=col3,pattern=north east lines] (v3) circle (.4cm); \draw[fill=white] (v4) circle (.4cm); \filldraw[pattern color=col3,pattern=north east lines] (v4) circle (.4cm); \end{tikzpicture} + \ldots  \end{gather*}
\caption{Illustration of the tropical loop equation. 
Each orange blob stands for a 1PI graph. 
The right-hand graphs have one 
purple pointed edge.
Cutting these  edges yields beaded graphs.
}
\label{fig:troplooprec}
\end{figure}

With this observation in hand, we are ready to prove 
the tropical loop equation~\eqref{eq:pde}, which we 
formally record as follows.
\begin{theorem}
\label{thm:pde}
The tropical effective action fulfills the PDE
\begin{align*} \mathcal P_D \, \Gamma^\tr = \left(1- \frac{\partial^2 \Gamma^\tr}{ \partial \varphi^2} \right)^{-1} - 1. \end{align*}
\end{theorem}
Figure~\ref{fig:troplooprec} illustrates this formula and the argument of the proof below.
The name \emph{tropical loop equation} emphasizes that the right-hand side 
contains only data of one loop order less than the left-hand side by stringing 1PI graphs into a loop.
\begin{proof}
From eq.~\eqref{eq:gammatr}, Lemma~\ref{lmm:PDdiag},
and Definition~\ref{def:Hepp}, we know that
\begin{align*} \mathcal P_D \Gamma^\tr = \sum_{G \in \mathcal G} \frac{\varphi^{n(G)} \prod_{v\in V_G} \lambda_{|v|}} {|\Aut' (G)|}\, 2\omega_D(G)\cdot \HE_D(G) = 2\sum_{G \in \mathcal G} \frac{\varphi^{n(G)} \prod_{v\in V_G} \lambda_{|v|}} {|\Aut' (G)|}\, \sum_{e\in E_G} \mathcal H_D(G\setminus e). \end{align*}
Let $\mathcal E$ be the set of isomorphism classes of \emph{edge-pointed 1PI graphs}, i.e.~pairs $(G,e)$ such that $G \in \mathcal G$ and $e \in E_G$.
We write $\Aut'(G,e)$ for the subgroup of $\Aut'(G)$ that fixes the edge $e$. Using the orbit-stabilizer theorem, we may switch the $1/\Aut'$ factor and the sum over $e$ in the equation above to get a generating function of Hepp-weighted edge-pointed 1PI graphs:
\begin{align*} \mathcal P_D \Gamma^\tr = 2 \sum_{(G,e) \in \mathcal E} \frac{\varphi^{n(G)} \prod_{v\in V_G} \lambda_{|v|}} {|\Aut' (G,e)|}\, \mathcal H_D(G\setminus e). \end{align*}
Figure~\ref{fig:troplooprec} illustrates this equation. The stated formula now follows by observing that the set of beaded graphs $\mathcal B$ and the set of edge-pointed 1PI graphs $\mathcal E$ are naturally isomorphic: If we cut the pointed edge of an edge-pointed 1PI graph and give a special colour to the two new legs, we get a beaded graph. If we connected the two special legs of a beaded graph and remember the position of the newly created edge, we get an edge-pointed 1PI graph. The respective notions of automorphisms are compatible. Further, we have $\mathcal H_D(G\setminus e) = \mathcal H (G_1) \cdot \mathcal H(G_2) \cdots$, where $G_1,G_2,\ldots$ are the 1PI components of $G\setminus e$.
The statement follows from Proposition~\ref{prop:beaded}.
\end{proof}

The PDE~\eqref{eq:pde} does not fix 
$\Gamma^\tr$ completely. 1PI graphs that have no edges (i.e.~vertex graphs) contribute terms of the form 
$ \lambda_k \frac{\varphi^k}{k!} $
to $\Gamma^\tr$. Such terms are mapped to $0$ by $\mathcal P_D$, as can be seen from Lemma~\ref{lmm:PDdiag}. 
Consistently, vertex graphs cannot be formed by connecting the two special legs of a beaded graph.
The boundary condition 
\begin{align*} \lim_{\hbar \rightarrow 0}\, \hbar \cdot \Gamma^\tr\left(D, \varphi {\hbar}^{-\frac12}, \hbar^{\frac{1}{2}} \lambda_3,\ldots,\hbar^{\frac{k-2}{2}} \lambda_k, \ldots\right) = \lim_{\hbar \rightarrow 0} \sum_{G \in \mathcal G} \hbar^{L(G)} \frac{\varphi^{n(G)} \prod_{v\in V_G} \lambda_{|v|}} {|\Aut' (G)|}\, \HE_D(G) = \lambda_k \frac{\varphi^k}{k!} \end{align*}
fixes the 1PI vertex contribution to $\Gamma^\tr$ to the value required by the definition in eq.~\eqref{eq:gammatr}.

The PDE~\eqref{eq:pde}, proven in Theorem~\ref{thm:pde}, exhibits a structure reminiscent of the Wetterich equation and, after some manipulation, also aligns well with the Polchinski equation~\cite{Polchinski:1983gv,Wetterich:1992yh}.\footnote{I thank Dario Benedetti and Kevin Costello for their illuminating explanations of these functional differential equations.} Both these are functional differential equations from the functional renormalization group framework. It would be very interesting to prove the PDE~\eqref{eq:pde} directly using functional renormalization group techniques instead of the perturbative route taken here.

\subsection{Perturbative tropicalized quantum field theory}
\label{sec:perttropqft}

In this article, we will not attempt an in-depth analysis of the PDE~\eqref{eq:pde}. It is not 
apparent if (and how) it could be solved analytically. 
In this section, we will explain how to solve it effectively in the ring of power series $\Q(D)[\![\varphi,\lambda_3,\lambda_4,\ldots]\!]$ in the formal variables $\varphi,\lambda_3,\lambda_4,\ldots$ over the field of rational functions in $D$. By Lemma~\ref{lmm:PDdiag}, the differential operator $\mathcal P_D$ is diagonal in the canonical basis of this ring over $\Q(D)$.
Further, it is easy to see that for 
connected graphs $G$ and generic $D$, we only have $\omega_D(G) =0$ if $G$ is a single-vertex graph.
So, by the definition of $\Gamma^\tr$ in eq.~\eqref{eq:gammatr},
\begin{align*} \Gamma^\tr = \sum_{k \geq 3} \varphi^k \frac{\lambda_k}{k!} + (\textit{terms that lie in the image of $\mathcal P_D$}). \end{align*}
Therefore, Theorem~\ref{thm:pde} implies that 
\begin{align} \label{eq:fixpoint} \Gamma^\tr = \sum_{k \geq 3} \varphi^k \frac{\lambda_k}{k!} + \mathcal P_D^{-1} \left( \left(1- \frac{\partial^2 \Gamma^\tr}{ \partial \varphi^2} \right)^{-1} - 1 \right), \end{align}
where $\mathcal P_D^{-1}$ is the canonical inverse of $\mathcal P_D$ in the monomial basis.
We remark that the fixed-point equation~\eqref{eq:fixpoint} is a spot-on manifestation 
of a \emph{Dyson--Schwinger equation} \cite[Chapter~19]{Bjorken:1965zz}. In particular, it can be easily put into the form of the `simple example' in \cite[\S 2]{Kreimer:2006ua}. This enables the study of this fixed-point equation using the Connes--Kreimer Hopf algebra of rooted trees~\cite{Connes:1998qv}. We leave such a study to future work.

Let $f$ be the operator $\Q(D)[\![\varphi,\lambda_3,\lambda_4,\ldots]\!] \rightarrow \Q(D)[\![\varphi,\lambda_3,\lambda_4,\ldots]\!]$
given by
\begin{align*} f(P)= \sum_{k \geq 3} \varphi^k \frac{\lambda_k}{k!} + \mathcal P_D^{-1} \left( \left(1- \frac{\partial^2 P}{ \partial \varphi^2} \right)^{-1} - 1 \right). \end{align*}
It is a small exercise in operator theory on such rings to
prove that with $P_{-1} = 0$, the iteration $P_{L+1} = f(P_L)$ converges to 
$\Gamma^\tr$ from eq.~\eqref{eq:gammatr}:
$\Gamma^\tr = \lim_{L\rightarrow \infty} P_L$ in the usual power series topology.
The $L$-th iteration $P_L$ agrees with the 
full series $\Gamma^\tr$ up to the $L$-th loop order.
The series expansion of $\Gamma^\tr$ can be computed easily using this iteration. The first few terms are
\begin{gather*} \Gamma^\tr =~ \frac{\varphi^3 \lambda_3}{3!} + \frac{\varphi^4 \lambda_4}{4!} + \frac{\varphi^5 \lambda_5}{5!} + \frac{\varphi^6 \lambda_6}{6!} + \ldots \\
   -\frac{\varphi \lambda_3}{D-2} - \frac{\varphi^2 \lambda_3^2}{D-4} - \frac{\varphi^2 \lambda_4}{2(D-2)}  -\frac{\varphi^3 \lambda_3^3}{D-6} -\frac{\varphi^3 \lambda_3 \lambda_4}{D-4} -\frac{\varphi^3 \lambda_5}{6(D-2)} + \ldots \\
+ \frac{\lambda_4}{2(D-2)^2}  +\frac{ \lambda_3^2 }{(D-3)(D-4)}  +\frac{\varphi \lambda_5}{2(D-2)^2}  + \frac{\varphi \lambda_3^3 (5D-24)}{(D-4)^2(D-6)}   +\frac{2\varphi \lambda_3 \lambda_4 (2D-5)}{(D-2)(D-3)(D-4)} + \ldots  \\
-\frac{\lambda_6}{6(D-2)^3} -2\frac{\lambda_3\lambda_5}{ (D-2)(D-3)(D-4) }  -\frac23\lambda_3^4\frac{- 2400 + 1412 D - 272 D^2 + 17 D^3}{ (D-4)^3(D-5)(D-6)(D-8) }    \\
-\lambda_4^2 \frac{32-25D+5D^2}{(D-2)^2(D-3)(D-4)(3D-8)}  -\lambda_3^2\lambda_4\frac{480 - 290 D+ 39D^2}{(D-2)(D-4)^2(D-6)(3D-10)}  +\ldots\\
+\ldots \end{gather*}
Here, the first line represents the $0$-loop (tree-level) terms in $\Gamma^\tr$ up to order $\varphi^6$. The first iteration, $P_0 = f(P_{-1}) = f(0)$, and all following $P_L$ with $L\geq 0$ contain these terms. From the next iteration, $P_1 = f(P_0)$ on, $P_L$ contains the correct first-loop terms, which are illustrated in the second line up to order $\varphi^3$. The third line contains the corresponding second loop order terms up to order $\varphi^1$. The three-loop contribution is shown up to order $\varphi^0$ (only vacuum contributions) in lines 4 and 5. Hence, iterating $f$ three times to get $P_3$, which agrees with $\Gamma^\tr$ up to third loop order, was sufficient to compute the displayed terms. Notice that each $\varphi^n$ coefficient  to the $L$-loop contribution is always a rational function in $D$ of total rational degree $-L$.

\subsection{Critical tropicalized quantum field theory}
\label{sec:crit}
The tropical loop equation~\eqref{eq:pde} shows interesting behaviour if we restrict ourselves to one coupling constant, i.e.~we put $\lambda_r=0$ for all $r \neq k$ and fix the dimension such that the $\lambda_k$ coupling becomes critical. Under these restrictions, the non-linear PDE~\eqref{eq:pde} reduces to a non-linear ODE. For instance, in the $\phi^4$ theory case, where $\lambda_r=0$ for all $r\neq 4$ and $D=4$, eq.~\eqref{eq:pde} reduces to
\begin{align} \label{eq:odephi4} \left( -4 + \varphi \frac{\partial}{\partial \varphi} \right) \Gamma_{\phi^4}^\tr = \left(1- \frac{\partial^2 \Gamma_{\phi^4}^\tr}{\partial \varphi^2} \right)^{-1} -1. \end{align}
Here, an interesting twist emerges:
The tropical loop equation becomes autonomous of the coupling, which suddenly only contributes via the boundary condition. It is unclear how to implement the perturbative 
approach from the last section: Now all 1PI graphs with four external legs lie in the kernel of the operator $-4 + \varphi \frac{\partial}{\partial \varphi}$, not just terms corresponding to single vertices. Without the regularization provided by a generic $D$, it is suddenly non-trivial to invert the operator $-4 + \varphi \frac{\partial}{\partial \varphi}$. The root of the problem lies in spurious perturbative divergences that appear in the double limit $\lambda_4 \rightarrow 0$ and $D\rightarrow 4$. The famous perturbative quantum field theory answer to this issue is \emph{renormalization}.
However, the tropical loop equation elegantly resolves these divergences on its own.
The coupling constant only enters via the boundary data that can conveniently also be interpreted as \emph{renormalization conditions} for the second-order ODE:
\begin{align*} \frac{\partial \Gamma_{\phi^4}^\tr}{\partial \varphi} \big|_{\varphi = 0} &= 0 \quad \text{and}\quad \frac{\partial^4 \Gamma_{\phi^4}^\tr}{\partial \varphi^4} \big|_{\varphi = 0} = \lambda_4. \end{align*}
The first condition enforces the $\varphi \leftrightarrow -\varphi$ symmetry of the theory.
The second one fixes the dressed, renormalized four-point function.
So, renormalization is naturally baked into the tropical loop equation. Interestingly, this procedure also effectively fixes the two-point function.  %
We will not attempt a detailed study of critical tropicalized quantum field theory here. 

The ODE~\eqref{eq:odephi4} is a special case of a known equation in the functional renormalization group community (see, e.g., equation (3) of \cite{Codello:2012sc} after substituting $\widetilde V_* \rightarrow \frac14 -\Gamma^\tr$, $d\rightarrow 4$ and $c_d =1$). In particular, numerical investigations found that the specific ODE~\eqref{eq:odephi4} has no non-trivial solutions that are finite for all real values of $\varphi$ \cite{Codello:2012sc,Hellwig:2015woa}. We leave the exploration of the consequences of this observation for tropicalized quantum field theory and further connections to functional renormalization group methods to future work.

\section{Global tropical sampling}
\label{sec:algo}
In \cite{Borinsky:2020rqs}, the author introduced a sampling algorithm that randomly generates points in $\mathbb P^{E_G}_{>0}$, weighted by evaluations of the tropicalized Symanzik polynomials $\mathcal U_G, \mathcal F_G$, tailored towards an efficient evaluation of Feynman integrals. Here, we will extend this idea to the entire quantum field theory and directly evaluate full perturbative coefficients in one go. Remarkably, and in contrast to~\cite{Borinsky:2020rqs}, the new global sampling algorithm that we will introduce runs in polynomial time. At sufficiently large loop order, it hence becomes significantly more efficient to evaluate the entire perturbative coefficient at once than to compute an individual Feynman integral contributing to it. In fact, due to the different runtime scalings, it will quickly become more efficient to evaluate the whole coefficient than \emph{any single Feynman integral at the same loop order}.

Recall the definition of the $n$-point correlation function in eqs.~\eqref{eq:defGamma}--\eqref{eq:defFeynman}, which imply:
\begin{gather} \notag \Gamma_{n}^1(p_1,\ldots,p_n) = \mu^{D+n(1-D/2)} \cdot \widehat \delta^{({D})}\left( \sum_{i=1}^n p_i\right) \cdot \widetilde \Gamma_{n}^1(p_1,\ldots,p_n)\\
\label{eq:Gamma1} \widetilde \Gamma_{n}^1(p_1,\ldots,p_n) = \sum_{\substack{G\in \mathcal G_n}} \frac{\prod_{v\in V_G} \lambda_{|v|}} {|\Aut (G)|} \frac{\Gamma(\omega_D(G)) } { (4\pi)^{L(G)\cdot D/2} } \int_{\mathbb P_{>0}^{E_G}} \frac{{\prod_{e\in E_G} x_e}} { \mathcal U_G(\bb x)^{D /2} \cdot \mathcal V_G(\bb x)^{\omega_D(G)} } \Omega\, , \end{gather}
where we set the deformation parameter back to $\xi=1$, as we are now interested in the original non-deformed quantum field theory. We aim to compute $\Gamma_{n}^1(p_1,\ldots,p_n)$. We will provide an algorithm that does so numerically. This algorithm relies crucially on the tropicalized quantum field theory in the $\xi\rightarrow 0^+$ limit
and its solution via the tropical loop equation~\eqref{eq:pde}.

In this section, to keep the notation manageable and for concreteness, we will mostly restrict ourselves to the $\phi^k$ case, where there is only one type of interaction. The algorithms and arguments presented here generalize to more complicated interactions, but at the cost of a large proliferation of indices for bookkeeping. To avoid this notational blowup, we will fix $k \geq 3$ and assume $\lambda_r = 0$ for all $r \neq k$.

Further, it is convenient to pass to indexing by loop order instead of the number of vertices. 
\begin{lemma}
\label{lmm:kregEdgeVert}
A connected, $k$-regular graph $G$ with $n$ legs and $L$ loops has
$\frac{2(L-1)+n}{k-2}$ vertices and
$\frac{(L-1)k+n}{k-2}$~edges.
\end{lemma}
\begin{proof}
Use the identities in the proof of Lemma~\ref{lmm:PDdiag} and 
$|v| =k$ for all $v\in V_G$.
\end{proof}

Let $\mathcal G_{L,n}^k$ be the set of isomorphism classes of $k$-regular (i.e.~$\phi^k$) 1PI graphs with $n$ legs and $L$ loops.
In the $k$-regular case, eq.~\eqref{eq:Gamma1} specializes to the following formula for the expansion of the $n$-point correlation function:
\begin{gather} \notag \widetilde \Gamma_{n}^1(p_1,\ldots,p_n) = \frac{1}{(4\pi)^{\frac{D}{2}(1-\frac{n}{2})}} \sum_{L \geq 1} \left( \frac{\lambda_k} { (4 \pi)^{\frac{D}{4} (k-2)} }\right)^{\frac{2(L-1)+n}{k-2}} \widetilde \Gamma_{L,n}^1(p_1,\ldots,p_n)\\
\label{eq:Gamma1Ln} \widetilde \Gamma_{L,n}^1(p_1,\ldots,p_n) = \sum_{\substack{G\in \mathcal G_{L,n}^k}} \frac{ \Gamma(\omega_{D}(G)) } {|\Aut (G)|} \int_{\mathbb P_{>0}^{E_G}} \frac{{\prod_{e\in E_G} x_e}} { \mathcal U_G(\bb x)^{D /2} \cdot \mathcal V_G(\bb x)^{\omega_{D}(G)} } \Omega\, . \end{gather}

The numerical evaluation algorithm will sample pairs 
of a graph $G$ and a corresponding point $\bb x \in \mathbb P^{E_G}_{>0}$.
We can interpret these pairs as metric graphs: Each homogeneous coordinate $x_e$ 
provides a length for the corresponding edge $e$ of $G$.
The pairs will be produced with probabilities that are close to the 
values of the summand and integrand in eq.~\eqref{eq:Gamma1Ln}. 
Before we dive into the inner workings of this sampling algorithm,
we will provide a conceptually helpful interpretation of $n$-point correlation functions 
as integrals over a moduli space of metric graphs.

\subsection{Moduli spaces of metric \texorpdfstring{$k$}{k}-regular graphs and their volumes}
\label{sec:modulimetric}

    In \cite{Mirzakhani:2006fta}, Maryam Mirzakhani
derived elegant recursive formulas to compute the volumes of moduli spaces of bordered Riemann surfaces (i.e., surfaces of given genus with specified boundary lengths). 
These volumes are given by integrating the Weil--Petersson form over the moduli space.

The discussion in the preceding sections can be viewed as an analogue of Mirzakhani's volume recursion. The present section aims to make this analogy a bit more precise: the moduli space of Riemann surfaces in Mirzakhani's setting is replaced here by a moduli space of graphs, and the Weil--Petersson volume form is replaced by a form determined by the tropicalized Symanzik polynomials that is natural from a quantum field theory perspective. Finally, in this correspondence, Mirzakhani's recursion equation is replaced by the tropical loop equation~\eqref{eq:pde}. %

Similar to Mirzakhani's recursion, which not only computes Weil--Petersson volumes but also enables the study of probability measures and sampling of random Riemann surfaces, the tropical loop equation in our framework provides a method to study specific probability measures on the moduli space of graphs. %
Remarkably, there is also an explicit, polynomial-time algorithm to sample random metric graphs accordingly. This sampling procedure allows for efficient numerical integration of Feynman-type perturbative expansions such as $n$-point correlation functions via a basic Monte Carlo workflow.

The moduli space of metric $k$-regular graphs $\MG^k_{L,n}$ with $L$ loops and $n$ legs 
is an \emph{orbispace} that combines all graphs with a (normalized) metric into one geometric object. 
Let 
$\mathcal G^k_{L,n}$ be the set of isomorphism classes of $k$-regular 1PI graphs with $n$ legs and $L$ loops.
As a set, $\MG^k_{L,n}$ is 
\begin{align*} \mathcal {MG}^k_{L,n} = \bigsqcup_{G \in \mathcal G^k_{L,n}} \mathbb P^{E_G}_{>0} / \Aut(G) \,, \end{align*}
where we take the disjoint union over all $k$-regular 1PI graphs $G$ with $n$ legs and $L$ loops and take the quotient with respect to the graphs' automorphism groups. As in Section~\ref{sec:pre}, the projective simplex $\mathbb P_{>0}^{E_G}$ is the positive part of $|E_G|-1$ dimensional real projective space parametrized by one homogeneous coordinate for each edge of $G$. The group $\Aut(G)$ acts on $\mathbb P^{E_G}_{>0}$ via the canonical action induced on $G$'s edges applied to these coordinates.
This action is not free. As an orbispace $\mathcal {MG}^k_{L,n}$ locally keeps track of the stabilizers of this action. 
It is a subspace of the moduli space of graphs and the (compact) moduli space of tropical curves of genus $L$ with $n$ marked points (see, e.g., \cite{chan2021moduli}): $\mathcal {MG}^k_{L,n} \subset \mathcal {MG}_{L,n} \subset \mathcal M^\tr_{L,n}$.

Each point in $\mathcal {MG}^k_{L,n}$ is represented by a pair $(G,\bb x)$ where $G$ is a $k$-regular 1PI graph and
$\bb x \in \mathbb P^{E_G}_{>0}$. We can think of such a point as a \emph{metric graph} of fixed total volume, where the homogeneous coordinate $x_e$ gives the length of the edge $e$. A differential $p$-form on $\mathcal {MG}^k_{L,n}$ is a collection of $p$-forms $\rho_G \in \Omega^p(\mathbb P^{E_G}_{>0})$ for each $G \in \mathcal G^k_{L,n}$ such that every $\rho_G$ is invariant under the action of $\Aut(G)$, i.e.~$\alpha^* \rho_G = \rho_G$ for all $\alpha \in \Aut(G)$.
On top of that, it is required that such a family of forms is compatible with a canonical gluing of the cells (see, e.g.,~\cite[\S 2.2]{Brown:2021umn}), but these additional conditions are irrelevant here, so we will not discuss them. 
Let $\rho = \{\rho_G\}_{G\in \mathcal G^k_{L,n}}$ be such a collection of top-forms. Orbifold integration over $ \mathcal {MG}^k_{L,n}$ is defined by
\begin{align*} \int_{ \mathcal {MG}^k_{L,n} } \rho = \sum_{ \substack{ G\in \mathcal G^k_{L,n} } } \frac{1}{|\Aut(G)|} \int_{\mathbb P_{>0}^{E_G}} \rho_G\, , \end{align*}
where we integrate over all cells in the moduli space and weight each point by the cardinality of the stabilizer that is attached to it.
A particularly interesting example of an appropriate form comes from $\phi^k$ quantum field theory. It is given by 
\begin{align*} \mu_{L,n} &= \{\mu_G\}_{G\in \mathcal G^k_{L,n}} &\text{ where }&& \mu_G &= \Gamma(\omega_{D}(G)) \frac{ {\prod_{e\in E_G} x_e}} { \mathcal U_G(\bb x)^{D /2} \cdot \mathcal V_G(\bb x)^{\omega_{D}(G)} } \, \Omega \, . \intertext{        We also define the tropicalized analogue of $\mu_{L,n}$ }  \mu_{L,n}^\tr &= \{\mu^\tr_G\}_{G\in \mathcal G^k_{L,n}} &\text{ where }&& \mu^\tr_G &= \frac{1}{\omega_{D}(G)} \frac{ {\prod_{e\in E_G} x_e}} { \mathcal U^\tr_G(\bb x)^{D /2} \cdot \mathcal V^\tr_G(\bb x)^{\omega_{D}(G)} } \, \Omega \, . \end{align*}
From the definition of $\mathcal U_G$ and $\mathcal V_G$, it follows that 
$\mu_G$ and $\mu_G^\tr$ are invariant under the action of $\Aut(G)$, so both 
families give rise to appropriate 
forms on $\mathcal {MG}^k_{L,n}$.

Combining this definition with \eqref{eq:Gamma1Ln} gives a compact formula for the $L$-loop $n$-point correlation function in $\phi^k$ quantum field theory as a moduli space integral:
\begin{align*} \widetilde \Gamma_{L,n}^1(p_1,\ldots,p_n) = \int_{\mathcal {MG}^k_{L,n}} \mu_{L,n}. \end{align*}

If the integral
of 
$\mu_{L,n}^\tr$ over 
$\mathcal {MG}^k_{L,n}$
exists, 
then we can define its normalized version,
\begin{align} \label{eq:tildemutrdef} \widetilde \mu_{L,n}^\tr = \frac{\mu_{L,n}^\tr}{Z^k_{D}(L,n)} \, , \text{ where } Z^k_{D}(L,n) = \int_{\mathcal {MG}^k_{L,n}} \mu_{L,n}^\tr \, , \end{align}
is an appropriate normalization factor. Positivity of the form $\mu_{L,n}^\tr$ then implies that 
$\widetilde \mu_{L,n}^\tr$ defines a probability measure on $\mathcal {MG}^k_{L,n}$. 

\begin{proposition}
\label{prop:finiteZ}
If for all $L \geq 1$ and $n \geq 2$, where 
$\frac{2(L-1)+n}{k-2}$
and $\frac{(L-1)k+n}{k-2}$ are integers,
\begin{align} \label{eq:omegaLn} \omega^k_D(L,n)   := \frac{(L-1)k+n}{k-2}-\frac{LD}{2} > 0\, ,  \end{align}
then 
the integral in eq.~\eqref{eq:tildemutrdef}
exist and $Z^k_{D}(L,n)$ is finite and non-negative 
for all $L \geq 0$ and $n \geq 2$.
\end{proposition}
By Lemma~\ref{lmm:tropapprox}, this proposition 
also implies that $\widetilde \Gamma_{L,n}^1(p_1,\ldots,p_n)$
is finite provided that the conditions on $m^2$ and $p_1,\ldots, p_n$ laid out in Section~\ref{sec:positivity} are fulfilled. We will postpone the proof of this proposition to the next section.

Using the results of the preceding sections, in particular the tropical loop equation, we will construct an efficient sampling algorithm that produces representatives of points $(G,\bb x) \in \mathcal {MG}^k_{L,n}$ with probability density $\widetilde \mu^\tr_{L,n}$. Here, efficient means that this algorithm runs in polynomial time and it has polynomially bounded memory requirements. 

In the remainder of this section, we will briefly describe the application of this algorithm to numerical integration over the moduli space of graphs and thereby the computation of perturbation theoretic quantities in quantum field theories such as $n$-point correlation functions.
It is convenient to define the residual function $f_{L,n}:\mathcal{MG}_{L,n}^k \rightarrow \R$ by 
\begin{align} \label{eq:fLn} f_{L,n}(G,x)= \Gamma(\omega_{D}(G)+1) \left( \frac{ \mathcal U_G^\tr(\bb x) } { \mathcal U_G(\bb x) } \right)^{D /2} \left( \frac{ \mathcal V^\tr_G(\bb x) } { \mathcal V_G(\bb x) } \right)^{\omega_{D}(G)}. \end{align}
As before with the top-forms $\mu_{L,n}$ and $\mu^\tr_{L,n}$,
the function $f$ descends to a function on $\mathcal{MG}_{L,n}^k$, because of the symmetries of the Symanzik polynomials 
under the group action of $\Aut(G)$.

Using~\eqref{eq:Gamma1Ln} and the definitions above, we find the following representation for the $n$-point correlation function in $\phi^k$-theory as an expectation value of the random variable given by $f_{L,n}$:
\begin{align*} \widetilde \Gamma_{L,n}^1(p_1,\ldots,p_n) = Z^k_D({L,n}) \cdot \int_{\mathcal {MG}^k_{L,n}} f_{L,n} \cdot \widetilde \mu_{L,n}^\tr = Z^k_D({L,n}) \cdot \mathbb E\left[ f_{L,n}\right]_{\mu_{L,n}^\tr} \, . \end{align*}
It follows from Lemma~\ref{lmm:tropapprox} and eq.~\eqref{eq:fLn} that
there exist constants $C_1(L,n)$ and $C_2(L,n)$
such that for all $(G,\bb x) \in \mathcal {MG}^k_{L,n}$, we have
$C_1(L,n) \leq f_{L,n}(G,\bb x) \leq C_2(L,n)$.
This provides all the necessary ingredients to evaluate 
the correlation function
$\widetilde \Gamma_{L,n}^1(p_1,\ldots,p_n)$ in one go as an integral over the moduli space of graphs:

\begin{proposition}
If $(G_1,\bb x_1), \ldots, (G_N,\bb x_N)$ are random pairs $(G_k,\bb x_k) \in \mathcal{MG}^k_{L,n}$ 
distributed independently following the probability density $\widetilde \mu^\tr_{L,n}$ on $\mathcal{MG}^k_{L,n}$,
then we can estimate 
$$
\widetilde 
\Gamma_{L,n}^1(p_1,\ldots,p_n)
\approx
\frac{Z^k_D(L,n)}{N}
\sum_{k=1}^N f_{L,n}(G_k,\bb x_k).
$$
For sufficiently large $N$, the error of this estimate follows a normal distribution with mean $0$
and variance $C({L,n})/N$, where $C({L,n})$ is a constant independent of $N$.
\end{proposition}
\begin{proof}
Follows from a standard argument based on the central limit theorem
(see, e.g., \cite[\S 2.2]{Borinsky:2020rqs} for an explanation in a related context).
\end{proof}

\subsection{Explicit recursions for tropicalized quantum field theory}
\label{sec:preproc}
The first and most crucial step in stating a global sampling algorithm for the measure $\widetilde \mu^\tr_{L,n}$ is to give an efficient way to compute the normalization constant $Z^k_D(L,n)$ for fixed $D\in \R$.
To do so, we first write $Z^k_D(L,n)$ as a Hepp bound weighted graph sum:
\begin{align*} Z^k_D({L,n}) = \sum_{G \in \mathcal G^k_{L,n}} \frac{ \mathcal H_D(G) }{|\Aut(G)|}, \end{align*}
which follows, e.g., from Proposition~\ref{prop:hepp} and the definitions of 
$ Z^k_D({L,n})$
and 
integration over $\mathcal {MG}^k_{L,n}$.
Here, we remark on another advantage of taking a global instead of a local graph-by-graph viewpoint: For a single graph $G$, it costs an with $L$ exponentially growing amount of runtime and memory to compute $\mathcal H_D(G)$ \cite[Corollary~3.10]{Panzer:2019yxl}. 
Here, in our global setup, we can compute the sum over all Hepp bounds $Z^k_D({L,n})$ recursively using the tropical loop equation in polynomial time, as we will now explain.

In Section~\ref{sec:tropqft}, we established that the numbers $ Z^k_D({L,n})$
 form the coefficients of the tropical effective action (see eq.~\eqref{eq:gammatr}).
In our, to the $\phi^k$ theory restricted case, we have
\begin{align} \label{eq:ZLn} \Gamma^\tr(\varphi, 0,\ldots, 0,\lambda_k, 0,\ldots)= \Gamma^\tr_{\phi^k} = \sum_{n \geq 0} \sum_{L \geq 0} Z^k_D({L,n}) \frac{\varphi^n}{n!} \lambda_k^{\frac{2(L-1)+n}{k-2}}, \end{align}
where we set $\lambda_r=0$ for all $r \neq k$ and the factor $n!$ accounts for passing from labelled legs to unlabelled legs, as we did in the argument for Theorem~\ref{thm:trop}. Recall Lemma~\ref{lmm:kregEdgeVert}, which implies that 
$\frac{2(L-1)+n}{k-2}$ is always an integer.

In analogy to $\mathcal G_{L,n}^k$, 
we define $\mathcal B_{L,n}^k$, the set of isomorphism classes of $k$-regular beaded graphs (see Definition~\ref{def:beaded}) with $n$ labelled legs and $L$ loops. We also require the two special legs to be distinctly labelled.  We define a second set of coefficients $B^k_D(L,n)\in \R$ via
\begin{align*} B^k_D(L,n) = \sum_{G \in \mathcal B_{L,n}^k} \frac{\mathcal H_D(G)}{|\Aut_{\mathcal B}(G)|}, \end{align*}
where $\Aut_{\mathcal B}(G) = \Aut(G) $ is the group of automorphisms of $G$ that fixes all legs.
From Proposition~\ref{prop:beaded} and Lemma~\ref{lmm:kregEdgeVert}, we know that 
\begin{align} \label{eq:BLn} \frac12 \left( 1- \frac{\partial^2 \Gamma^\tr_{\phi^k}}{\partial \varphi^2} \right)^{-1} - \frac12 = \sum_{n \geq 2} \sum_{L \geq 0} B^k_D({L,n}) \frac{\varphi^{n-2}}{2 (n-2)!} \lambda_k^{\frac{2(L-1)+n}{k-2}} , \end{align}
where again the additional factor $2(n-2)!$ accounts for fixing the legs.

By the tropical loop equation, as, e.g., given in eq.~\eqref{eq:troploopphik}, we have
\begin{align} \label{eq:PDEBZ} \mathcal P^{\phi^k}_D \Gamma^\tr_{\phi^k} = \left( 1- \frac{\partial^2 \Gamma^\tr_{\phi^k}}{\partial \varphi^2} \right)^{-1} -1. \end{align}
By Lemma~\ref{lmm:PDdiag}, $\mathcal P^{\phi^k}_D \varphi^n \lambda_k^m = 2 \omega^k_D(L,n) \, \varphi^n \lambda_k^m$, where $m=\frac{2(L-1)+n}{k-2}$ and 
with $\omega_D^k(L,n)$ from eq.~\eqref{eq:omegaLn}. 
We have $\omega_D^k(L,n) = \omega_D(G)$ for all $G \in \mathcal G_{L,n}^k$, by Lemma~\ref{lmm:kregEdgeVert}.
By combining this with eqs.~\eqref{eq:ZLn}, \eqref{eq:BLn}, and \eqref{eq:PDEBZ}, we get a relation between the coefficients $B^k_D(L,n)$ and $Z^k_D(L,n)$,
\begin{align} \label{eq:Zrec} Z^k_D(L,n) &= \frac{ B^k_D(L-1,n+2)}{2 \omega^k_D(L,n)} \text{ for all } n\geq 0 \text{ and } L \geq 1, \end{align}
which allows to compute $Z^k_D(L,n)$ for all $n$ if 
$B^k_D(L-1,n')$ is known for all values of $n'$.
The boundary condition for the tropical loop equation 
fixes the $L=0$ case, i.e.~the values $Z^k_D(0,n)$ for all $n \geq 0$.
In our $\phi^k$ case, we find
$Z^k_D(0,k) = 1$ and $Z^k_D(0,m) = 0$ for all $m \neq k$.
The identity 
$$
\frac12
\left(
1-
\frac{\partial^2 \Gamma^\tr_{\phi^k}}{\partial \varphi^2}
\right)^{-1}
-
\frac12 
=
\frac12 
\frac{\partial^2 \Gamma^\tr_{\phi^k}}{\partial \varphi^2}
+
\frac12 
\frac{\partial^2 \Gamma^\tr_{\phi^k}}{\partial \varphi^2}
\left(
\left(
1-
\frac{\partial^2 \Gamma^\tr_{\phi^k}}{\partial \varphi^2}
\right)^{-1}
-\frac12 \right)
$$
translates into the following recursion equation that computes $B^k_D(L,n)$ for all $L \geq 0$ and $n \geq 2$,
provided that $Z^k_D(L',n')$ is known for all $L'\leq L$ and $n' \leq n$:
\begin{align} \label{eq:Brec} B^k_D(L,n) = Z^k_D(L,n) + \sum_{n' = 0}^{n-2} \sum_{L' =0}^{L} \binom{n-2}{n'} Z^k_D(L',n'+2) B^k_D(L-L',n-n'). \end{align}
The recursion terminates, because $Z^k_D(0,n)=B^k_D(0,n)=0$ for all $n < k$.

Combining eqs.~\eqref{eq:Zrec} and \eqref{eq:Brec} gives a loop-by-loop recursion and thereby an efficient algorithm to compute $Z^k_D(L,n)$.
Counting the number of steps needed to compute $Z^k_D(L,n)$ shows that the required computing time scales as $\bigO(L^2(L+n)^2)$ while the required memory scales as $\bigO(L(L+n))$.
\begin{theorem}
For fixed $D$, the value of $Z^k_D(L,n)$ can be computed in 
polynomial time and memory in $L$ and $n$ using the  recursion equations \eqref{eq:Zrec} and \eqref{eq:Brec}.
\end{theorem}

\begin{proof}[Proof of Proposition~\ref{prop:finiteZ}]
Eqs.~\eqref{eq:Zrec} and \eqref{eq:Brec}
together with the stated inequality \eqref{eq:omegaLn} imply that $Z^k_D(L,n)$ and $B^k_D(L,n)$ are finite and non-negative 
for all $L \geq 0$ and $n\geq 2$.
\end{proof}

\subsection{Polynomial-time sampling of metric graphs}
\label{sec:poly}

\begin{figure}

\def\rad{.3}
\def\rud{1.1}
  \begin{subfigure}[b]{\textwidth}
\begin{align*} \begin{tikzpicture}[every path/.style={thick},baseline={(0,-.1)}] \coordinate (v1) at (0,0); \coordinate (v11) at (-\rud,0); \coordinate (v12) at (\rud,0); \coordinate (v1i1) at ({-\rad-\rud},0); \coordinate (v1i2) at ({\rad+\rud},0); \filldraw[pattern color=col3,pattern=north east lines] (v1) circle(\rud); \node[fill=white,circle] at (v1) { {\parbox{.9cm}{\centering $L$ \\  $n$}} }; \end{tikzpicture} \quad \longleftarrow \begin{tikzpicture}[every path/.style={thick},baseline={(0,-.1)}] \coordinate (v1) at (0,0); \coordinate (v11) at (-\rud,0); \coordinate (v12) at (\rud,0); \coordinate (v1i1) at ({-\rad-\rud},0); \coordinate (v1i2) at ({\rad+\rud},0); \filldraw[pattern color=col1,pattern=dots] (v1) circle(\rud); \node[fill=white,circle] at (v1) { {\parbox{.9cm}{\centering $L-1$ \\  $n+2$}} }; \draw[col4,double] (v11) -- (v1i1); \draw[col4,double] (v12) -- (v1i2); \fill (v11) circle(2pt); \fill (v12) circle(2pt); \def\handleLen{.5} \coordinate (M) at (0,{\rad+\rud}); \coordinate (C1a) at ($ (v1i1) + (-\handleLen,0) $); \coordinate (C1b) at ($ (M) + (-\rud,0) $); \coordinate (C2a) at ($ (M) + (\rud,0) $); \coordinate (C2b) at ($ (v1i2) + (\handleLen,0) $); \draw[col4,double] (v1i1) .. controls (C1a) and (C1b) .. (M) .. controls (C2a) and (C2b) .. (v1i2); \end{tikzpicture} \longleftarrow \quad \begin{tikzpicture}[every path/.style={thick},baseline={(0,-.1)}] \coordinate (v1) at (0,0); \coordinate (v11) at (-\rud,0); \coordinate (v12) at (\rud,0); \coordinate (v1i1) at ({-\rad-\rud},0); \coordinate (v1i2) at ({\rad+\rud},0); \filldraw[pattern color=col1,pattern=dots] (v1) circle(\rud); \node[fill=white,circle] at (v1) { {\parbox{.9cm}{\centering $L-1$ \\  $n+2$}} }; \draw[double] (v11) -- (v1i1); \draw[double] (v12) -- (v1i2); \fill (v11) circle(2pt); \fill (v12) circle(2pt); \end{tikzpicture} \end{align*}
    \caption{Illustration of the $L> 0$ case of Algorithm~\ref{algo:1PI}.}
    \label{fig:algo1PI}
  \end{subfigure}

  \begin{subfigure}[b]{\textwidth}
\begin{align*} \def\shift{2.8} \begin{tikzpicture}[every path/.style={thick},baseline={(0,-.1)}] \coordinate (v1) at (0,0); \coordinate (v11) at (-\rud,0); \coordinate (v12) at (\rud,0); \coordinate (v1i1) at ({-\rad-\rud},0); \coordinate (v1i2) at ({\rad+\rud},0); \filldraw[pattern color=col1,pattern=dots] (v1) circle(\rud); \node[fill=white,circle] at (v1) { {\parbox{.9cm}{\centering $L$ \\  $n$}} }; \draw[double] (v11) -- (v1i1); \draw[double] (v12) -- (v1i2); \fill (v11) circle(2pt); \fill (v12) circle(2pt); \end{tikzpicture} \quad \longleftarrow \quad \begin{tikzpicture}[every path/.style={thick},baseline={(0,-.1)}] \coordinate (v1) at (0,0); \coordinate (v2) at ({\shift},0); \coordinate (v11) at (-\rud,0); \coordinate (v12) at (\rud,0); \coordinate (v21) at ({\shift-\rud},0); \coordinate (v22) at ({\shift+\rud},0); \coordinate (v1i1) at ({-\rad-\rud},0); \coordinate (v2i2) at ({\shift+\rud+\rad},0); \filldraw[pattern color=col3,pattern=north east lines] (v1) circle(\rud); \filldraw[pattern color=col1,pattern=dots] (v2) circle(\rud); \node[fill=white,circle] at (v1) { {\parbox{.9cm}{\centering $L'$ \\  $n'+2$}} }; \node[fill=white,circle] at (v2) { {\parbox{.9cm}{\centering $L-L'$ \\  $n-n'$}} }; \draw[double] (v11) -- (v1i1); \draw[double] (v12) -- (v21); \draw[double] (v22) -- (v2i2); \fill (v11) circle(2pt); \fill (v12) circle(2pt); \fill (v21) circle(2pt); \fill (v22) circle(2pt); \end{tikzpicture} \end{align*}
    \caption{Illustration of the $(L',n')$ case in Algorithm~\ref{algo:beaded}.}
    \label{fig:algobeaded}
  \end{subfigure}

\caption{
Illustrations of the algorithms.
The orange diagonal hatchings indicate 1PI graphs.
The blue dotted pattern stands for  beaded graphs. }
\label{fig:algo-illustration}
\end{figure}
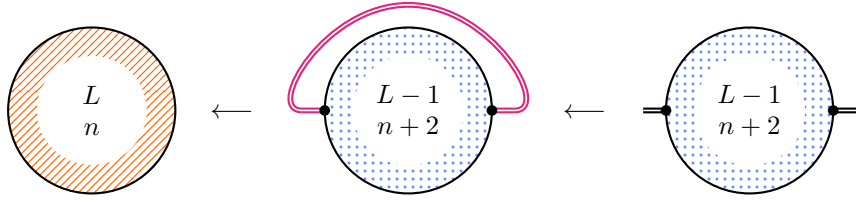
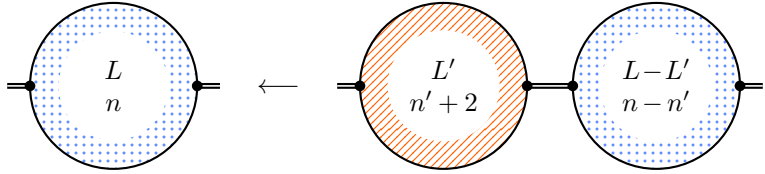

\begin{figure}

\end{figure}

We are now ready to state the global sampling algorithm 
that produces pairs $(G,\bb x)$ distributed as the 
measure $\widetilde \mu^\tr_{L,n}$ defined in eq.~\eqref{eq:tildemutrdef}. 
The strategy is to follow the steps of the recursive computation of the normalization factors $Z^k_D(L,n)$ and $B^k_D(L,n)$ from the last section and lift these recursions to integrals over graph moduli spaces.
The algorithm consists of two routines 
that invoke each other recursively.
The first routine randomly generates a metric 1PI graph.
The second one randomly generates a metric beaded graph.

\begin{algorithm}[Sampling a 1PI graph]\mbox{}
\label{algo:1PI}
\vspace{5pt}

\textbf{Input:} Integers $L \geq 0$ and $n \geq 0$ such that $Z^k_D(L,n) > 0$.

\textbf{Output:} A pair $(G, \bb \xiz)$  of a 1PI graph $G \in \mathcal G^k_{L,n}$ and an element $\bb \xiz \in [0,1]^{E_G}$.
\vspace{5pt}

First, if $L =0$ and $n=k$, then return 
the pair $(V_k,\bb \xiz)$ formed by the $k$-vertex graph $V_k$ and $\bb \xiz = \emptyset$.
If $L \geq 1$, then follow the steps:
\begin{enumerate}
\item Randomly sample a pair $(H, \bb \xiz^H)$ of a beaded graph $H \in \mathcal B^k_{L-1,n+2}$ 
with $L-1$ loops, $n+2$ legs and an element $\bb \xiz^H \in [0,1]^{E_H}$ using Algorithm~\ref{algo:beaded}.
\item Form a new graph $G$ from $H$ by adding a new edge $e'$ that connects both special legs.
\item Uniformly draw a random real number $\lambda \in [0,1]$ and set $\kappa = \lambda^{1/\omega_D(L,n)}$ with $\omega_D(L,n)$ given as in eq.~\eqref{eq:omegaLn}.
\item Form an element $\bb \xiz \in [0,1]^{E_G}$ by setting $\xiz_e = \kappa \cdot \xiz_e^H$ for $e\in E_H$ and
$\xiz_{e'} = \kappa$.
\item Return $(G,\bb \xiz)$.
\end{enumerate}

\end{algorithm}

Figure~\ref{fig:algo1PI} illustrates the essence of Algorithm~\ref{algo:1PI}: gluing together the special legs of a beaded graph. The new graph's edge parameters are rescaled such that the new edge's (depicted in purple) parameter is the largest.

\begin{algorithm}[Sampling a beaded graph]\mbox{}
\label{algo:beaded}
\vspace{5pt}

\textbf{Input:} Integers $L \geq 0$ and $n \geq 2$  such that $B^k_D(L,n) > 0$.

\textbf{Output:} A pair $(G, \bb \xiz)$  of a graph $G \in \mathcal B^k_{L,n}$ and an element $\bb \xiz \in [0,1]^{E_G}$.
\vspace{5pt}

Randomly pick one of $1+(L+1)(n-1)$ outcomes with different probabilities:

\begin{itemize}
\item
Pick outcome $\star$ with probability $\frac{Z_D^k(L,n)}{B_D^k(L,n)}$.

\item
For each of the $(L+1)(n-1)$ pairs of integers $(L',n')$ where $0\leq L' \leq L$ and $0\leq n' \leq n-2$,
pick outcome $(L',n')$ with probability 
$\binom{n-2}{n'} \frac{Z_D^k(L',n'+2) B_D^k(L-L',n'-n)}{B_D^k(L,n)}$.
\end{itemize}

Depending on the picked outcome, follow different steps:
\vspace{5pt}

\textbf{In case of outcome $\star$}:
\begin{enumerate}
\item
Randomly sample a pair $(G,\bb \xiz)$ of a graph $G \in \mathcal G^k_{L,n}$ and an element $\xiz \in [0,1]^{E_G}$ using Algorithm~\ref{algo:1PI}. 
\item
Declare the first two legs of $G$ to have the special colour, promoting $G$ to a beaded graph.
\item
Return $(G, \bb \xiz)$ as output.
\end{enumerate}

\textbf{In case of outcome $(L',n')$} follow the steps:
\begin{enumerate}
\item
Randomly sample a pair $(A,\bb \xiz^A)$ of a 1PI graph $A \in \mathcal G^k_{L',n'+2}$ with $L'$ loops and $n'+2$ legs and an element $\bb \xiz^A \in [0,1]^{E_A}$ using Algorithm~\ref{algo:1PI}. 

\item
Randomly sample a pair $(B, \bb \xiz^B)$ of a beaded graph $B \in \mathcal B^k_{L - L',n-n'}$ 
with $L-L'$ loops and $n-n'$ legs 
and $\bb \xiz^B \in [0,1]^{E_B}$ using (the current) Algorithm~\ref{algo:beaded}.

\item
Randomly and uniformly sample a subset $S$ of size $n'$ of $\{3,\ldots,n\}$.

\item
Assign the numbers in $S$ to all legs of $A$ labelled $\geq 3$ in the unique, order-preserving way.

\item
Analogously, assign the numbers $\{3,\ldots,n\} \setminus S$ to the non-special legs of $B$.
\item
Form a new graph $G$ starting with the disjoint union of $A$ and $B$.
Connect the second leg of $A$ with the first special leg of $B$ via a new edge $e'$.
Declare the first leg of $A$ to be special such that $G$ becomes a beaded graph with 
legs labelled by $\{1,\ldots,n\}$.
\item
Randomly and uniformly sample a real number $\lambda \in [0,1]$.
\item
Assign the values in $\bb \xiz^A$ and $\bb \xiz^B$ to 
the corresponding edge in $G$ and
the $\lambda$ to the new edge $e'$,
$\xiz_{e'} = \lambda$. This way, form an element $\bb \xiz \in [0,1]^{E_{G}}$.
\item
Return the pair $(G, \bb \xiz)$.
\end{enumerate}

\end{algorithm}

The initial discrete sampling step leading to the different outcomes $\star$ and $(L',n')$ is well-posed by eq.~\eqref{eq:Brec} and Proposition~\ref{prop:finiteZ}, which guarantee that the total probability of picking any outcome is 1.
Figure~\ref{fig:algobeaded} illustrates the essence of Algorithm~\ref{algo:beaded}: concatenating a 1PI graph with a beaded graph by gluing one of the beaded graph's special legs to one of the 1PI graph's legs and declaring another one of the 1PI graph's legs to be the new special leg.

\begin{theorem}\mbox{}
\label{thm:sampling1PI}
We assume that the conditions of Proposition~\ref{prop:finiteZ} are fulfilled.
For given input $L, n \geq 0$ such that $Z^k_D(L,n) >0$, Algorithm~\ref{algo:1PI} outputs the pair $(G,\bb \xiz)$, where $G\in \mathcal G^k_{L,n}$ and $\bb \xiz \in [0,1]^{E_G}$ with probability density
\begin{align*}   \frac{1}{Z^k_D(L,n)} \frac{1}{|\Aut(G)|} \frac{\prod_{e\in E_G} \dd \xiz_e}{\mathcal U_G^\tr(\bb \xiz)^{D/2}}. \end{align*}
\end{theorem}
\begin{theorem}
\label{thm:samplingbeaded}
We assume that the conditions of Proposition~\ref{prop:finiteZ} are fulfilled.
For given input $L \geq 0$, $n \geq 2$ such that $B^k_D(L,n) >0$, Algorithm~\ref{algo:beaded} outputs the pair $(G,\bb \xiz)$ where $G\in \mathcal B^k_{L,n}$ and $\bb \xiz \in [0,1]^{E_G}$ with probability density
\begin{align*}   \frac{1}{B^k_D(L,n)} \frac{1}{|\Aut(G)|} \frac{\prod_{e\in E_G} \dd \xiz_e}{\mathcal U_G^\tr(\bb \xiz)^{D/2}}. \end{align*}
\end{theorem}

While studying the proofs of these theorems below, it may be instructive to occasionally return to Section~\ref{sec:hepp} to review the necessary integral manipulations that are used heavily here.

\begin{lemma}
\label{lmm:samp1PI}
For input pairs $(L,n)$ where $L=0$, Theorem~\ref{thm:sampling1PI} holds.
If Theorem~\ref{thm:samplingbeaded} holds for the input pair $(L-1,n+2)$  where $L \geq 1$, then 
 Theorem~\ref{thm:sampling1PI} holds for $(L,n)$.
\end{lemma}

\begin{proof}
For $L=0$, Theorem~\ref{thm:sampling1PI} holds trivially,
because Algorithm~\ref{algo:1PI} only outputs the vertex graph without any edge parameters.

Let $f$ be a test function on the set of pairs $(G,\bb \xiz)$ where $G\in \mathcal G^k_{L,n}$ and $\bb \xiz \in [0,1]^{E_G}$. We will compute the expectation value of $f$ evaluated on a random pair $(G,\bb \xiz)$ that is produced by Algorithm~\ref{algo:1PI}.
For input pairs $(L,n)$ with $L \geq 1$, Algorithm~\ref{algo:1PI} invokes Algorithm~\ref{algo:beaded} with the input pair $(L-1,n+2)$.
By assumption of the validity of Algorithm~\ref{algo:beaded} for this input and by translating the listed steps of Algorithm~\ref{algo:1PI} into an integral formula, we get the expectation value 
\begin{align*} \mathbb E[f] &= \frac{1}{B^k_D(L,n)} \sum_{H\in \mathcal{B}^k_{L-1,n+2}} \frac{1}{|\Aut(H)|} \int_{0}^1 \dd \lambda \int_{[0,1]^{E_H} } f(G(H), \bb \xiz(\lambda^{1/\omega_D^k(L,n)}, \bb \xiz^H)) \cdot \frac{\prod_{e\in E_H} \dd \xiz_e^H }{\mathcal U_H^\tr(\bb \xiz^H)^{D/2}}, \end{align*}
where $G(H)$ is the 1PI graph obtained from $H$ as described in Step~2 of Algorithm~\ref{algo:1PI},
and $\bb \xiz(\kappa, \bb \xiz^H)$ is the element of $[0,1]^{E_{G(H)}}$ obtained from $\bb \xiz^H$ and $\kappa = \lambda^{1/\omega^k_D(L,n)}$ as described in Step~4.
From eq.~\eqref{eq:Zrec}, and the change of variables 
to $\kappa$, we get 
\begin{align*} \mathbb E[f] &= \frac12 \frac{1}{Z^k_D(L,n)} \sum_{H\in \mathcal{B}^k_{L-1,n+2}} \frac{1}{|\Aut(H)|} \int_{0}^1 \frac{\dd \kappa}{\kappa} \kappa^{\omega_D^k(L,n)} \int_{[0,1]^{E_H} } f(G(H), \bb \xiz(\kappa, \bb \xiz^H)) \cdot \frac{\prod_{e\in E_H} \dd \xiz_e^H }{\mathcal U_H^\tr(\bb \xiz^H)^{D/2}}. \end{align*}
Let $\mathcal {EG}_{L,n}^k$ be the set of pairs $(G,e)$ of a graph $G \in \mathcal G^k_{L,n}$ and an edge $e \in E_G$ up to isomorphism, and $\Aut(G,e) \subset \Aut(G)$ the subgroup of automorphisms of $G$ that fix $e$. As in the proof of Theorem~\ref{thm:pde}, we may pass from summing over all
beaded graphs $\mathcal{B}^k_{L-1,n+2}$ to edge-pointed 
graphs in $\mathcal{EG}^k_{L,n}$.
We get an extra factor of $2$ by accounting for the additional $2$-fold symmetry of edge-pointed graphs, where we can change the order of the two legs that were glued together:
\begin{align*} \mathbb E[f] &= \frac{1}{Z^k_D(L,n)} \sum_{(G,e')\in \mathcal{EG}^k_{L,n}} \frac{1}{|\Aut(G,e')|} \int_{0}^1 \frac{\dd \kappa}{\kappa} \kappa^{\omega_D^k(L,n)} \int_{[0,1]^{E_H} } f(G, \bb \xiz^{e'}(\kappa, \bb \xiz^H)) \cdot \frac{\prod_{e\in E_H} \dd \xiz_e^H }{\mathcal U_H^\tr(\bb \xiz^H)^{D/2}}, \end{align*}
where we made the dependence of $\bb \xiz = \bb \xiz^{e'}(\kappa, \bb \xiz^H)$ on the edge that is added to $H$ to form $G$ explicit, and we keep in mind that $H$ is equal to $G\setminus e'$.
Using the orbit stabilizer theorem allows us to pass from edge-pointed graphs to 1PI by summing over all edges of $G$:
\begin{align*} \mathbb E[f] &= \frac{1}{Z^k_D(L,n)} \sum_{G\in \mathcal{G}^k_{L,n}} \frac{1}{|\Aut(G)|} \sum_{e'\in E_G} \int_{0}^1 \frac{\dd \kappa}{\kappa} \kappa^{\omega_D^k(L,n)} \int_{[0,1]^{E_H} } f(G, \bb \xiz^{e'}(\kappa, \bb \xiz^H)) \cdot \frac{\prod_{e\in E_H} \dd \xiz_e^H }{\mathcal U_H^\tr(\bb \xiz^H)^{D/2}}. \end{align*}
The function $\bb \xiz^{e'}: (0,1) \times [0,1]^{E_H} \rightarrow [0,1]^{E_G}, (\kappa, \bb \xiz^H) \mapsto \bb \xiz^{e'}$ is given explicitly 
by $\xiz_e^{e'} = \kappa \xiz^H_e$ for all $e\in E_H$ and $\xiz_{e'}^{e'} = \kappa$. The image of this injective 
map is the set $C_{e'}(E_G) \subset [0,1]^{E_G}$ 
as defined before Lemma~\ref{lmm:kruskal}.
Using this lemma,
$\omega_D^k(L,n) = |E_G| - L(G)D/2$,
$\deg \mathcal U_H = L-1$, 
and passing to the $\bb \xiz^{e'}$ coordinates leads to
\begin{align*} \mathbb E[f] &= \frac{1}{Z^k_D(L,n)} \sum_{G\in \mathcal{G}^k_{L,n}} \frac{1}{|\Aut(G)|} \sum_{e'\in E_G} \int_{C_{e'}(E_G)} f(G, \bb \xiz^{e'}) \cdot \frac{\prod_{e\in E_G} \dd \xiz^{e'}_e }{\mathcal U_G^\tr(\bb \xiz^{e'})^{D/2}}\\
&= \frac{1}{Z^k_D(L,n)} \sum_{G\in \mathcal{G}^k_{L,n}} \frac{1}{|\Aut(G)|} \int_{[0,1]^{E_G}} f(G, \bb \xiz) \cdot \frac{\prod_{e\in E_G} \dd \xiz_e }{\mathcal U_G^\tr(\bb \xiz)^{D/2}}. \end{align*}
We used that the sets $C_{e'}(E_G)$ for $e' \in E_G$ partition $[0,1]^{E_G}$ up to nullsets as in Corollary~\ref{cor:HDcube}.
\end{proof}

\begin{proof}[Proof of Theorem~\ref{thm:samplingbeaded}]
We prove the statement by induction.
Fix a pair $(L,n)$ with $L \geq 0$ and $n \geq 2$. We assume that Theorem~\ref{thm:samplingbeaded} holds for all pairs $(L',n')$ with $L' \leq L$ and either $n' < n$ or $L' < L$.
We will invoke Algorithm~\ref{algo:beaded} for such a pair $(L,n)$.

Let $f$ be a test function on the set of pairs $(G,\bb \xiz)$ where $G\in \mathcal B^k_{L,n}$ and $\bb \xiz \in [0,1]^{E_G}$. Consider the probability path associated with the outcome $\star$ first. While executing the steps $\star$ of Algorithm~\ref{algo:beaded}, we need to invoke Algorithm~\ref{algo:1PI} with the input pair $(L,n)$. By Lemma~\ref{lmm:samp1PI}, we may use Theorem~\ref{thm:sampling1PI} for this input, as Theorem~\ref{thm:samplingbeaded} is assumed to be valid for all pairs $(L',n')$ where $L'<L$. The contribution to the expectation value $\mathbb E[f]$ is simply 
\begin{align*} E_\star = \frac{1}{Z^k_D(L,n)} \sum_{G \in \mathcal G^k_{L,n}} \frac{1}{|\Aut(G)|} \int_{[0,1]^{E_G}} f(G,\bb \xiz) \frac{\prod_{e\in E_G} \dd \xiz_e}{\mathcal U_G^\tr(\bb \xiz)^{D/2}}, \end{align*}
where, in the argument to $f$, we interpret the 1PI graph $G$ as a beaded graph by giving the special colour to the first two legs.

Next, we look into the case of the probability paths for the outcomes $(L',n')$ of Algorithm~\ref{algo:beaded}.
These paths will be chosen with probability 
$ \binom{n-2}{n'} {Z_D^k(L',n'+2) B_D^k(L-L',n'-n)}/{B_D^k(L,n)}$.

Under the induction hypothesis, and by Lemma~\ref{lmm:samp1PI}, we may assume the validity of Theorem~\ref{thm:sampling1PI} for all pairs $(L',n')$ with $L'\leq L$. 
The value of $Z_D^k(0,r)$ is zero if $r< k$, because there are no $k$-regular tree graphs with fewer than $k$ legs. Hence, we can assume that either $L' \geq 1$ or $n'+2 \geq k \geq 3$, as the other cases will only be invoked with probability $0$.
The probability path $(L',n')$ invokes Algorithm~\ref{algo:beaded} recursively with the input $L-L'$ and $n-n'$.
Because either $L' \geq 1$ or $n' +2 \geq 3$, we may assume the validity of this algorithm for the respective input
under the induction hypothesis.

For given $L',n'$, a 1PI graph $A \in {\mathcal G}^k_{L',n'+2}$,
a beaded graph $B \in {\mathcal B}^k_{L-L',n-n'}$
and a subset $S\subset \{3,\ldots,n\}$,
denote $G(A,B,S)$ the graph formed by Step~6 of the $(L',n')$ part of Algorithm~\ref{algo:beaded}.
For given $\bb \xiz^A \in [0,1]^{E_A}, \bb \xiz^B \in [0,1]^{E_B}$ and $\lambda \in [0,1]$,
write $\bb \xiz(\bb \xiz^A,\bb \xiz^B,\lambda)$
for the point in the cube $[0,1]^{E_G}$ 
obtained via Step~8 of this part of the algorithm.
The contribution to the expectation value $\mathbb E[f]$ is hence
\begin{gather*} E_{L',n'}= \frac{1}{Z^k_D(L',n'+2)} \sum_{A \in {\mathcal G}^k_{L',n'+2}} \frac{1}{|\Aut(A)|} \frac{1}{B^k_D(L-L',n-n')} \sum_{B \in {\mathcal B}^k_{L-L',n-n'}} \frac{1}{|\Aut(B)|} \times \\
\frac{1}{\binom{n-2}{n'}} \sum_{ \substack{ S \subset \{3,\ldots,n\}\\
|S| = n' } } \int_{[0,1]^{E_A}} \int_{[0,1]^{E_B}} \int_0^1 f \left(G(A,B,S),\bb \xiz(\bb \xiz^A,\bb \xiz^B, \lambda) \right) \frac{\prod_{e\in E_A} \dd \xiz_e^A} { \mathcal U_A^\tr(\bb \xiz^A)^{D/2} } \frac{\prod_{e\in E_B} \dd \xiz_e^B} { \mathcal U_B^\tr(\bb \xiz^B)^{D/2} } \dd \lambda\,. \end{gather*}
Because the legs of $A$ and $B$ are fixed,
$|\Aut(A)|\cdot|\Aut(B)|=|\Aut(G(A,B,S))|$.
Because all spanning trees of $G(A,B,S)$ have 
to pass through the bridge connecting the subgraphs 
$A$ and $B$, we also find that 
$\mathcal U^\tr_A(\bb \xiz^A) \mathcal U^\tr_B(\bb \xiz^B) = \mathcal U^\tr_{G(A,B,S)}(\bb \xiz). $
Hence, this contribution equals 
\begin{gather*} E_{L',n'}= \frac{1}{Z^k_D(L',n'+2)} \frac{1}{B^k_D(L-L',n-n')} \frac{1}{\binom{n-2}{n'}} \sum_{A \in {\mathcal G}^k_{L',n'+2}} \sum_{B \in {\mathcal B}^k_{L-L',n-n'}} \sum_{ \substack{ S \subset \{3,\ldots,n\}\\
|S| = n' } } \\
\times \frac{1}{|\Aut(G(A,B,S))|} \int_{[0,1]^{E_{G(A,B,S)}}} f \left(G(A,B,S),\bb \xiz \right) \frac{\prod_{e\in E_{G(A,B,S)}} \dd \xiz_e} { \mathcal U_{G(A,B,S)}^\tr(\bb \xiz)^{D/2} }\,. \end{gather*}

Let $\mathcal B^k_{L,n}(L',n') \subset \mathcal B^k_{L,n}$ be the set of beaded graphs with $L$ loops and $n$ legs,
 which have at least one bridge, and which
decompose into a 1PI graph with $L'$ loops and $n'+2$ 
legs and a smaller beaded graph if the bridge that is closest to the first leg is removed.
The construction under Step~6 of the $(L',n')$ part of Algorithm~\ref{algo:beaded} provides a bijection 
between $\mathcal B^k_{L,n}(L',n')$ and triples
$(A,B,S)$ where $A\in \mathcal G^k_{L',n'+2}, B \in \mathcal B^k_{L-L',n-n}$ and $S \subset \{3,\ldots,n\}$ with $|S|=n'-2$. Hence,
\begin{gather*} E_{L',n'}= \frac{1}{Z^k_D(L',n'+2)} \frac{1}{B^k_D(L-L',n-n')} \frac{1}{\binom{n-2}{n'}} \sum_{G \in \mathcal B^k_{L,n}(L',n')} \int_{[0,1]^{E_{G}}} \frac{ f \left(G,\bb \xiz \right) }{|\Aut(G)|} \frac{\prod_{e\in E_{G}} \dd \xiz_e} { \mathcal U_{G}^\tr(\bb \xiz)^{D/2} }\,. \end{gather*}
Summing over all outcomes, $\star$ and $(L',n')$, weighted by their respective probabilities, results in 
\begin{align*} \mathbb E[f] &= \frac{Z^k_D(L,n)}{B^k_D(L,n)} E_\star + \sum_{L'=0}^L \sum_{n'=0}^{n-2} \binom{n-2}{n'} \frac{Z_D^k(L',n'+2) B_D^k(L-L',n'-n)}{B_D^k(L,n)} E_{L',n'} \\
&= \sum_{G \in \mathcal G^k_{L,n}} \int_{[0,1]^{E_{G}}} \frac{ f \left(G,\bb \xiz \right) }{|\Aut(G)|} \frac{\prod_{e\in E_{G}} \dd \xiz_e} { \mathcal U_{G}^\tr(\bb \xiz)^{D/2} } \\
&+ \sum_{L'=0}^L \sum_{n'=0}^{n-2} \sum_{G \in \mathcal B^k_{L,n}(L',n')} \int_{[0,1]^{E_{G}}} \frac{ f \left(G,\bb \xiz \right) }{|\Aut(G)|} \frac{\prod_{e\in E_{G}} \dd \xiz_e} { \mathcal U_{G}^\tr(\bb \xiz)^{D/2} }\,. \end{align*}
Every beaded graph in $\mathcal B^k_{L,n}$
corresponds either to a 1PI graph
in $\mathcal G^k_{L,n}$ or 
a beaded graph of type $L',n'$
in $\mathcal B^k_{L,n}(L',n')$.
Therefore,
\begin{align*} \mathbb E[f] &= \frac{1}{B^k_D(L,n)} \sum_{G \in \mathcal B^k_{L,n}} \int_{[0,1]^{E_{G}}} \frac{ f \left(G,\bb \xiz \right) }{|\Aut(G)|} \frac{\prod_{e\in E_{G}} \dd \xiz_e} { \mathcal U_{G}^\tr(\bb \xiz)^{D/2} }\, , \end{align*}
which proves the claim.
\end{proof}

\begin{proof}[Proof of Theorem~\ref{thm:sampling1PI}]
Follows immediately from Lemma~\ref{lmm:samp1PI}
and Theorem~\ref{thm:samplingbeaded}.
\end{proof}

Algorithms~\ref{algo:1PI}
and~\ref{algo:beaded}
only produce samples of pairs $(G, \bb \xiz)$ with cubical edge coordinates $\bb \xiz$. We can directly map these samples to samples that are distributed as the normalized tropical measure $\widetilde \mu^\tr_{L,n}$ as given in eq.~\eqref{eq:tildemutrdef}:

\begin{theorem}
There is a canonical many-to-one map from $(0,1)^{E_G}$ 
to $\mathbb P_{>0}^{E_G}$ by reinterpreting $\bb \xiz \in (0,1)^{E_G}$ as a set of homogeneous coordinates.
Sampling a pair $(G,\bb \xiz)$ via Algorithm~\ref{algo:1PI} and mapping the resulting element $\bb \xiz$ to $\mathbb P_{>0}^{E_G}$ via this canonical map produces samples that are distributed as the normalized tropical measure $\widetilde \mu^\tr_{L,n}$ on the moduli space of $k$-regular graphs.
\end{theorem}
\begin{proof}
Follows from Theorem~\ref{thm:sampling1PI} and the argument in Proposition~\ref{prop:hepp}'s proof.
\end{proof}

The recursions in Algorithms~\ref{algo:1PI} and~\ref{algo:beaded} always reduce the total number of edges in the next lower rank of the recursion tree. The initial discrete random sampling step of Algorithm~\ref{algo:beaded} takes constant time after a preprocessing step of runtime and memory requirements $\bigO(Ln)$ using the alias method \cite[\S 3.4.1]{knuth1997taocp2}. The sampling step of the subset $S$ in Algorithm~\ref{algo:beaded} takes at most runtime $n_{\mathrm{max}}$, where $n_{\mathrm{max}}$ is the maximal number of legs of a graph encountered in the course of the computation. This number is smaller than $2|E_G|+n$, where $|E_G|$ and $n$ are the number of edges and legs of the output graph. In total, we get that the number of computational steps grows at most as $|E_G|(2|E_G|+n)$. In our $k$-regular case, $|E_G|$ is proportional to the number of loops of $G$. Hence, we can record the following on the runtime of 
Algorithms~\ref{algo:1PI}
and~\ref{algo:beaded}:
\begin{theorem}
\label{thm:runtime}
The runtime of Algorithms~\ref{algo:1PI}
and~\ref{algo:beaded}
is at most of the order of $\bigO(L(L+n))$.
\end{theorem}

Running Algorithms~\ref{algo:1PI} and~\ref{algo:beaded} requires the computation of the normalization constants $Z_D^k(L,n)$ and $B_D^k(L,n)$ as described in Section~\ref{sec:preproc}. As we explained there, these can be computed in runtime $\bigO(L^2(L+n)^2)$, and they require  $\bigO(L(L+n))$ memory.  We thus have a polynomial-time and polynomial-memory algorithm for the sampling of the normalized tropical measure $\widetilde \mu^\tr_{L,n}$ on the moduli space of $k$-regular graphs.

\section{Some illustrative applications}
\label{sec:applications}
In this section, we will briefly discuss two concrete applications of the global sampling algorithm from the last section.  
A proof-of-concept \texttt{C++} implementation of the algorithm is attached with the ancillary files to the 
arXiv version of this article and available on GitHub.\footnote{\url{https://github.com/michibo/amplitrop}} All computations can be reproduced using this 
implementation.
We leave a full, in-depth interpretation of the resulting data to future~work.

\subsection{Application to massive \texorpdfstring{$\phi^3$}{phi3} theory in \texorpdfstring{$D=3$}{D=3}}
\label{sec:applicationphi3}

The algorithm introduced in the last section samples graphs proportional to their contribution to the tropicalized perturbative expansion.  
For that to be immediately applicable without further modifications, we need 
all contributing (tropicalized) Feynman integrals to be finite.  For scalar massive quantum field theories, this only holds in low dimensions $D$. For instance, we can use the algorithm to compute the $n$-point correlation functions in $\phi^3$ theory in $D=3$ for $n \geq 2$, as these $n$-point functions are finite and have no subdivergences:
the conditions for Proposition~\ref{prop:finiteZ} are fulfilled with $k=3$ and $D=3$. 
Such computations of finite quantities in low-dimensional super-renormalizable quantum field theories have interesting phenomenological applications.
See, for instance, \cite{Sberveglieri:2023mzy} for a recent exposition.

As an illustration of the algorithm's capabilities, we evaluate the 1PI $3$-point correlation function with zero momentum transfer and with the mass parameter $m^2=1$. That means, we compute the coefficients 
$ \widetilde \Gamma_{L,3}^1(0,0,0) $
from~\eqref{eq:Gamma1Ln}
in the special case $k=3$, $n=3$, $p_1=p_2=p_3 = 0$,
and $m^2/\mu^2 = 1$ in~\eqref{eq:UF},
\begin{align} \notag \widetilde \Gamma_{3}^1(0,0,0) = (4\pi)^{\frac{3}{4}} \sum_{L \geq 1} \left( \frac{\lambda_3} { (4 \pi)^{\frac{3}{4}} }\right)^{2(L-1)+3} \widetilde \Gamma_{L,3}^1(0,0,0), \end{align}
which are explicitly given by
\begin{align} \begin{aligned} \label{eq:phi3integral} \widetilde \Gamma_{L,3}^1(0,0,0) &= \sum_{\substack{G\in \mathcal G_{L,3}^3}} \frac{ \Gamma(\omega_{3}(G)) } {|\Aut (G)|} \int_{\mathbb P_{>0}^{E_G}} \frac{{\prod_{e\in E_G} x_e}} { \mathcal U_G(\bb x)^{3 /2} \cdot \mathcal V_G(\bb x)^{\omega_{3}(G)} } \Omega\,\\
&= Z^3_3(L,3) \cdot \int_{\mathcal{MG}^3_{L,3}} \Gamma(\omega_{3}(G)+1) \left( \frac{ \mathcal U_G^\tr(\bb x) }{ \mathcal U_G(\bb x) } \right)^{3/2} \left( \frac{ \mathcal V_G^\tr(\bb x) }{ \mathcal V_G(\bb x) } \right)^{\omega_{3}(G)} \mu^\tr_{L,3}\, , \end{aligned} \end{align}
where, as we are in the zero momentum case,
$\mathcal V_G(\bb x) = m^2/\mu^2 \sum_{e\in E_G} x_e = \sum_{e\in E_G} x_e$.

\begin{table}
\centering
\begin{tabular}{llll}
\toprule
$L$ & samples & $\widetilde \Gamma_{L,3}^1(0,0,0)$ & time/h \\
\midrule
1 & $1 \cdot 10^{11}$ & $4.431109 \cdot 10^{-1}$ $\pm$ $1.1 \cdot 10^{-6}$ & 2 \\
2 & $1 \cdot 10^{11}$ & $1.047191 \cdot 10^{0}$ $\pm$ $5.9 \cdot 10^{-6}$ & 3 \\
3 & $1 \cdot 10^{11}$ & $2.902190 \cdot 10^{0}$ $\pm$ $3.6 \cdot 10^{-5}$ & 3 \\
4 & $1 \cdot 10^{11}$ & $8.877142 \cdot 10^{0}$ $\pm$ $2.7 \cdot 10^{-4}$ & 4 \\
5 & $1 \cdot 10^{11}$ & $2.920635 \cdot 10^{1}$ $\pm$ $2.4 \cdot 10^{-3}$ & 6 \\
6 & $1 \cdot 10^{11}$ & $1.019640 \cdot 10^{2}$ $\pm$ $2.4 \cdot 10^{-2}$ & 6 \\
7 & $1 \cdot 10^{11}$ & $3.748502 \cdot 10^{2}$ $\pm$ $3.3 \cdot 10^{-1}$ & 7 \\
8 & $1 \cdot 10^{11}$ & $1.440633 \cdot 10^{3}$ $\pm$ $2.1 \cdot 10^{0}$ & 8 \\
9 & $1 \cdot 10^{11}$ & $5.787627 \cdot 10^{3}$ $\pm$ $2.2 \cdot 10^{1}$ & 9 \\
10 & $1 \cdot 10^{11}$ & $2.399101 \cdot 10^{4}$ $\pm$ $1.4 \cdot 10^{2}$ & 7 \\
11 & $1 \cdot 10^{11}$ & $1.074911 \cdot 10^{5}$ $\pm$ $2.6 \cdot 10^{3}$ & 12 \\
12 & $1 \cdot 10^{11}$ & $4.760706 \cdot 10^{5}$ $\pm$ $1.2 \cdot 10^{4}$ & 13 \\
13 & $1 \cdot 10^{11}$ & $2.235488 \cdot 10^{6}$ $\pm$ $1.0 \cdot 10^{5}$ & 15 \\
14 & $1 \cdot 10^{11}$ & $1.000354 \cdot 10^{7}$ $\pm$ $3.3 \cdot 10^{5}$ & 16 \\
15 & $1 \cdot 10^{11}$ & $5.464614 \cdot 10^{7}$ $\pm$ $4.0 \cdot 10^{6}$ & 16 \\
16 & $1 \cdot 10^{11}$ & $2.859931 \cdot 10^{8}$ $\pm$ $3.4 \cdot 10^{7}$ & 17 \\
17 & $1 \cdot 10^{11}$ & $1.156947 \cdot 10^{9}$ $\pm$ $3.6 \cdot 10^{7}$ & 20 \\
18 & $1 \cdot 10^{11}$ & $8.861573 \cdot 10^{9}$ $\pm$ $1.6 \cdot 10^{9}$ & 20 \\
19 & $1 \cdot 10^{11}$ & $7.159013 \cdot 10^{10}$ $\pm$ $3.6 \cdot 10^{10}$ & 23 \\
20 & $1 \cdot 10^{11}$ & $2.776484 \cdot 10^{11}$ $\pm$ $5.2 \cdot 10^{10}$ & 24 \\
\bottomrule
\end{tabular}
\caption{$3$-point function computation in massive $\phi^3$ theory in $D=3$. }
\label{tab:phi3_loop_values}

\end{table}

Up to loop order 20, Table~\ref{tab:phi3_loop_values} lists the number of samples taken, the estimate for the coefficient, and the computing time in hours on all cores of a computer.  The computation ran on various nodes of Perimeter Institute's computing cluster. All these nodes were equipped with an \texttt{AMD EPYC 7532 32-Core Processor}. 
At 20 loops, the memory requirement of the algorithm is of the order of a couple of kilobytes, which is negligible by 2025 standards. So, we did not record the memory usage.

Table~\ref{tab:phi3_loop_values} illustrates the favourable runtime scaling of the algorithm with increasing loop order. In the regime up to $L=20$, the algorithm's runtime appears to scale as $\bigO(L^{\frac32})$, faster than the estimates obtained for Theorem~\ref{thm:runtime}.

The implementation uses the Eigen \cite{eigenweb} linear algebra library to compute the Cholesky decomposition of the graph Laplacian whose determinant equals $\mathcal U_G$ (see, e.g., \cite[\S 4.2]{Borinsky:2023jdv} for an exposition of fast Feynman integrand evaluation methods). For larger values of $L$, the computation of this decomposition at each sampled point will become significant. It scales as $\bigO(L^3)$.

We verified the computation for the first couple of loop orders using \texttt{feyntrop}~\cite{Borinsky:2023jdv}.
However, we have to give a cautionary warning: as we will discuss in the next section, the reader should be careful while interpreting the data in Table~\ref{tab:phi3_loop_values} at larger loop orders $L$. Especially, the error intervals might be overly optimistic, as these are also estimated via a Monte Carlo process.

We will leave the interpretation and explicit application of such computations to future work. 

\subsection{Bounds on the estimates' relative error}
\label{sec:error}
In Table~\ref{tab:phi3_loop_values}, we observe that the error bands for the obtained estimates grow significantly with the loop order. So, even though the global sampling algorithm and all preprocessing steps run in polynomial time, we need to sample a rapidly growing number of points with the
loop order to get a fixed relative accuracy for our $n$-point function estimate. In fact, the number of points we need to sample might grow exponentially or more. So, by obtaining a polynomial-time algorithm for the sampling of points over the moduli space with a bounded fluctuation term, we have not necessarily achieved a polynomial-time algorithm to evaluate the original perturbative coefficient to fixed relative accuracy.

    This section briefly explains how to study this error term and the relative accuracy qualitatively by providing hard bounds on it. Suppose we want to estimate the expectation value of a random variable (i.e.~a function) $f$ on the moduli space $\mathcal{MG}^k_{L,n}$  via sampling from the tropical probability measure $\widetilde \mu^\tr_{L,n}$. That means, we aim to estimate the integral
\begin{align*} \mathbb E[f] = \int_{ \mathcal {MG}^k_{L,n} } f(G,\bb x)\, \widetilde \mu^\tr_{L,n}\,. \end{align*}
Provided that the variance, $\mathrm{Var}[f ] = \mathbb E[f^2] - (\mathbb E[f ])^2$ is finite, we can estimate $\mathbb E[f ]$ by averaging over many evaluations of $f$ on samples $(G_1, \bb x_1), \ldots , (G_N , \bb x_N )$, which are drawn independently using the measure $\widetilde \mu^\tr_{L,n}$. This is the typical Monte Carlo approach:
\begin{align*} \mathbb E[f] \approx \frac{1}{N} \sum_{i=1}^N f(G_i,\bb x_i) \, . \end{align*}
The statistical error of this estimate follows a Gaussian distribution with mean $0$ and variance $\mathrm{Var}[f]/N$ by the central
limit theorem.
In our setup, we will only consider functions $f$ that are bounded due to applications of Lemma~\ref{lmm:tropapprox}. So, the variance is always finite, the central limit theorem holds, and the overall Monte Carlo process is guaranteed to work. 

However, to get concrete bounds on the \emph{relative} error of the estimate, we also need a well-behaved lower bound for $f$. Suppose that we find (e.g., by using an argument similar to the one for Lemma~\ref{lmm:tropapprox}) that $f$ is bounded by constants $C_1, C_2 > 0$ such that $C_1 \leq f (G, \bb x) \leq C_2$ for all
representatives of metric graphs $(G, \bb x) \in \mathcal{MG}^k_{L,n}$. Then the expected relative error is
\begin{align} \label{eq:relerror} \delta[f] = \frac{\sqrt{\mathrm{Var}[f]/N}}{\mathbb E[f]}= \frac{1}{\sqrt{N}} \sqrt{ \frac{\mathbb E[f^2]}{\mathbb E[f]^2} -1 } \leq \frac{1}{\sqrt{N}} \frac{\sqrt{\mathbb E[f^2]}}{\mathbb E[f]} \leq \frac{1}{\sqrt{N}} \frac{C_2}{C_1} \, . \end{align}
Notice that by the Cauchy--Schwarz inequality, $\mathbb E[f^2] \geq \mathbb E[f ]^2$. The 
obtained bound on the relative error in the last step is governed by the ratio of the lower and upper bounds on $f$.

As a concrete example, we will discuss the relative error for the estimates of our $\phi^3$ theory $3$-point function in $D = 3$ above.
To get concrete error bounds, we need to analyze the integrand over the moduli space $\mathcal{MG}^3_{L,3}$ in eq.~\eqref{eq:phi3integral}. 
From the argument for Lemma~\ref{lmm:tropapprox}, we get the bounds,
$$
\frac{1}{|\mathcal T_G|}
\leq 
\frac{
\mathcal U_G^\tr(\bb x)
}{
\mathcal U_G(\bb x)
}
\leq 1\, ,
$$
where $|\mathcal T_G|$ is the number of spanning trees of $G$, and 
$$
\frac{1}{|E_G|} \leq 
\frac{
\mathcal V_G^\tr(\bb x)
}{
\mathcal V_G(\bb x)
}
=
\frac{
\max_{e\in E_G} x_e
}
{
\sum_{e\in E_G} x_e
}
\leq 1\,.
$$
In both cases, the upper bounds look promising. Unfortunately, the lower bounds can get tiny.

Recall (e.g.~from Lemma~\ref{lmm:kregEdgeVert}) that for a $L$ loop 3-legged 3-regular graph $G$, we have $|E_G| = 3L$ and $\omega_3(G) = \frac32 L$. 
Moreover, a graph has at most as many spanning trees 
as it has edge-subgraphs. So, $|\mathcal T_G| \leq 2^{|E_G|}$.
For the integration kernel
$$f(G, \bb x) = 
\Gamma(\omega_{3}(G)+1)
\left(
\frac{
\mathcal U_G^\tr(\bb x)
}{
\mathcal U_G(\bb x)
}
\right)^{3/2}
\left(
\frac{
\mathcal V_G^\tr(\bb x)
}{
\mathcal V_G(\bb x)
}
\right)^{\omega_{3}(G)}\, ,
$$
we find, hence, that in total 
$\delta[f] \leq (24L)^{\frac32 L}$ by eq.~\eqref{eq:relerror}.

Unfortunately, this bound grows more than factorially with $L$. The source of this huge growth rate is the term 
$ \left( { \mathcal V_G^\tr(\bb x) }/{ \mathcal V_G(\bb x) } \right)^{\omega_{3}(G)} $
in the moduli space integrand~\eqref{eq:phi3integral}. In particular, the exponent $\omega_3(G)$ that grows with the loop order leads to a super-exponentially growing bound.

In the practical computation listed in Table~\ref{tab:phi3_loop_values}, the error intervals do not grow as much as this theoretical bound on $\delta[f]$ indicates. In fact, the error only seems to grow exponentially with the loop order. Still, the growth of the variance obstructs the study of loop orders in massive scalar $\phi^3$ theory in $D=3$ beyond $\approx 20$ using the present approach without any modifications.

    Salvatori discusses variance reduction techniques in the context of sampling of surface integrals \cite[\S 5]{Salvatori:2025oib}. Plausibly, importing those ideas to the present context can lead to improvements of our estimates and our theoretical bounds in their errors.

    Further,  the exponent $\omega_3(G)$ only grows because we chose to study a super-renormalizable quantum field theory. In a renormalizable quantum field theory, this exponent would remain constant, and we can expect a much better theoretical bound on the relative error.

    In the next section, we will therefore sketch a way to go forward and also perform computations in quantum field theories in the critical dimension using a slight modification of the algorithm from Section~\ref{sec:algo}.

\subsection{Renormalization and the positive Hepp bound}
\label{sec:renormalization}

In this section, we introduce an extension of the algorithm from Section~\ref{sec:algo} that allows us to evaluate certain perturbative objects in higher-dimensional quantum field theories where divergences appear,
and the conditions for Proposition~\ref{prop:finiteZ} are not fulfilled anymore. The standard tool to manage divergences in quantum field theory is renormalization. 
Because we only consider a massive scalar field here, all relevant divergences are UV divergences.

\begin{definition}
\label{def:posHepp}
The positive Hepp bound $\mathcal H_D^{\mathrm{pos}}(G)$ is defined for 
each graph $G$ by setting 
$\mathcal H_D^{\mathrm{pos}}(G) = 1$ if $G$ is a graph without edges,
by
$\mathcal H_D^{\mathrm{pos}}(H) = \prod_{k} \mathcal H_D^{\mathrm{pos}}(G_k) $
for a non-1PI graph $H$ with 1PI components $G_1,\ldots,G_n$ where $n \geq 2$
and otherwise if $G$ is a 1PI graph, 
\begin{align*}  \HE_D^{\mathrm{pos}}(G) = \begin{cases} \frac{1}{\omega_D(G)} \sum_{e\in E_G} \HE_D^{\mathrm{pos}}(G\setminus e) & \text{ if } \omega_D(G) > 0 \\
0 & \text{ else} \, , \end{cases} \end{align*}
which recursively defines the positive Hepp bound for all graphs.
\end{definition}
The positive Hepp bound is not a rational function in $D$, and it can have quite complicated singularities in $D$.
We remark that the positive Hepp bound  is a renormalized version of the original Hepp bound from Definition~\ref{def:Hepp}. This can be proven using the Connes-Kreimer interpretation of renormalization \cite{Connes:1998qv,Connes:1999yr} (see also \cite{Panzer:2012gp}). 

For fixed $D$, we define the generating function of the positive Hepp bound as a renormalized version of the tropical effective action~\eqref{eq:gammatr}:
\begin{align*} \begin{gathered}  \Gamma^\tr_{\mathrm{pos}} = \sum_{G \in \mathcal G} \frac{\varphi^{n(G)} \prod_{v\in V_G} \lambda_{|v|}} {|\Aut' (G)|}\, \HE^{\mathrm{pos}}_D(G) \, . \end{gathered} \end{align*}

The argument in Section~\ref{sec:tropqft} can be applied transparently to give a fixed-point equation for the positive Hepp bound. The result is
\begin{align*} \Gamma^\tr_{\mathrm{pos}} = \sum_{k \geq 3} \varphi^k \frac{\lambda_k}{k!} + \mathcal P_D^{-1} \circ (\id - K_D) \circ \left( \left(1- \frac{\partial^2 \Gamma^\tr_{\mathrm{pos}}}{ \partial \varphi^2} \right)^{-1} - 1 \right), \end{align*}
where $\id$ is the identity operator, and 
$K_D$ is the operator that acts on basis monomials that correspond to 1PI graphs by mapping 
$K_D ( \varphi^{n(G)} \prod_{v\in V_G} \lambda_{|v|} ) = \varphi^{n(G)} \prod_{v\in V_G} \lambda_{|v|}$ if $\omega(G) < 0$ and $K_D ( \varphi^{n(G)} \prod_{v\in V_G} \lambda_{|v|} ) = 0$ if $\omega(G) \geq 0$. In our setup, $K_D$ takes the role of the operator isolating the divergent contribution that is typically used in the BPHZ formalism. In fact, it is possible 
(but very inefficient) to compute $\HE_D^{\mathrm{pos}}(G)$ from $\HE_D(G)$ using the BPHZ formalism.

Replacing $\HE_D(G)$ with $\HE^{\mathrm{pos}}_D(G)$ in Section~\ref{sec:algo} yields an algorithm that samples pairs $(G,\bb x)$ of metric graphs weighted by the positive Hepp bound $\HE^{\mathrm{pos}}_D(G)$. By Definition~\ref{def:posHepp},
it is obvious that $\HE_D(G) = \HE^{\mathrm{pos}}_D(G)$ if $\omega_D(\gamma) > 0$ for all subgraphs $\gamma \subset G$. That means, for graphs without subdivergences,
the positive Hepp bound and the original Hepp bound are equal.
By definition, the positive Hepp bound is 
always non-negative and finite. So, this modified sampling algorithm is not obstructed by issues of potential divergences.

\subsection{Application to the primitive \texorpdfstring{$\phi^4$}{phi4} theory \texorpdfstring{$\beta$}{beta} function}
\label{sec:applicationphi4}

We will apply this modified algorithm based on the positive Hepp bound to compute a particularly interesting quantity in scalar quantum field theory: the primitive contribution to the scalar $\phi^4$ theory $\beta$ function \cite{Broadhurst:1985tld,Broadhurst:1995km,Schnetz:2008mp}. This quantity has long served as a benchmark for new computational techniques. While it is known analytically up to 7 loops, at loop order 8, new number-theoretic phenomena emerge, which pose obstructions for an evaluation using existing analytic methods \cite{Brown:2010bw}. 
With Hepp-bound-based methods, Kompaniets and Panzer 
estimated this primitive contribution up to 11 loops \cite{Kompaniets_2017}. 
Recently, Balduf used tropical sampling to study the primitive contribution to the $\phi^4$ theory $\beta$ function up to loop order 18 \cite{Balduf:2023ilc} (see also \cite{Balduf:2024gvv,Balduf:2024njk}). Up to 17 loops, these computations were later confirmed by Favorito and the author \cite{Borinsky:2025ncs}.
A folklore conjecture states that the primitive contribution to the minimally subtracted $\phi^4$ beta function becomes dominant at large loop order \cite{McKane:2018ocs} (see also \cite{Dunne:2021lie} for an illuminating discussion).

A $4$-point graph in $\phi^4$ theory, $G \in \mathcal G_{L,4}^4$, is 
primitive if it is logarithmically divergent for $D=4$ (meaning that $\omega_4(G)=0$) and free of subdivergences (meaning that $\omega_4(\gamma) > 0$ for all subgraphs $\gamma \subset G$).
The primitive contribution to the $\beta$ function in  $\phi^4$ theory is (see~\cite[Eq.~(B1)]{Kompaniets_2017}),
\begin{align*} \beta^{\mathrm{prim}}_{L+1} = 2\sum_{ \substack{ G \in \mathcal G_{L,4}^4\\
} }\frac{ \Theta( G \text{ is primitive} ) }{|\Aut(G)|} \int_{\mathbb P_{>0}^{E_G}} \frac{\prod_{e\in E_G} x_e }{\mathcal U_G(\bb x)^{2}} \, \Omega \, , \end{align*}
where $ \Theta( G \text{ is primitive} ) $ is $1$ if $G$ is primitive and $0$ otherwise. 
We can use the positive versions of the Algorithms~\ref{algo:1PI} and \ref{algo:beaded} to 
estimate $\beta^{\mathrm{prim}}_{L+1}$ by writing it in the form
\begin{align*} \beta^{\mathrm{prim}}_{L+1} = 2 Z^{4,\mathrm{pos}}_{4}(L,4) \int_{\mathcal M_{L,4}^4} \Theta( G \text{ is primitive} )\, \left( \frac{ \mathcal U_G^\tr(\bb x)}{\mathcal U_G(\bb x)} \right)^2 \, \widetilde \mu^{\tr,\mathrm{pos}}_{L,4} \, , \end{align*}
where $ Z^{4,\mathrm{pos}}_{4}(L,4) $ and $ \widetilde \mu^{\tr,\mathrm{pos}}_{L,4}$ are the positive versions 
of the normalization factor and measure 
$ Z^{4}_{4}(L,4) $ and $ \widetilde \mu^{\tr}_{L,4}$
from Sections~\ref{sec:modulimetric} and \ref{sec:preproc}.

\begin{table}
\centering
\begin{tabular}{cccccc}
\toprule
$L$ & samples & $N_{\text{Prim}}$ & $\beta^{\mathrm{prim}}_{L+1}$ & $\beta H^{\mathrm{prim}}_{L+1}$ & time/h \\
\midrule
3 & $1.10 \cdot 10^{10}$ & $1.87 \cdot 10^{9}$ & $1.442497 \cdot 10^{1}$ $\pm$ $3.0 \cdot 10^{-4}$ & $1.679980 \cdot 10^{2}$ $\pm$ $3.5 \cdot 10^{-3}$ & 0 \\
4 & $1.10 \cdot 10^{10}$ & $1.31 \cdot 10^{9}$ & $1.244281 \cdot 10^{2}$ $\pm$ $3.5 \cdot 10^{-3}$ & $3.432005 \cdot 10^{3}$ $\pm$ $8.9 \cdot 10^{-2}$ & 1 \\
5 & $1.10 \cdot 10^{10}$ & $1.28 \cdot 10^{9}$ & $1.698163 \cdot 10^{3}$ $\pm$ $5.5 \cdot 10^{-2}$ & $1.135437 \cdot 10^{5}$ $\pm$ $3.0 \cdot 10^{0}$ & 1 \\
6 & $1.10 \cdot 10^{10}$ & $1.18 \cdot 10^{9}$ & $2.412932 \cdot 10^{4}$ $\pm$ $9.1 \cdot 10^{-1}$ & $3.958005 \cdot 10^{6}$ $\pm$ $1.1 \cdot 10^{2}$ & 1 \\
7 & $1.10 \cdot 10^{10}$ & $1.10 \cdot 10^{9}$ & $3.709545 \cdot 10^{5}$ $\pm$ $1.6 \cdot 10^{1}$ & $1.509371 \cdot 10^{8}$ $\pm$ $4.3 \cdot 10^{3}$ & 1 \\
8 & $1.10 \cdot 10^{10}$ & $1.04 \cdot 10^{9}$ & $6.062108 \cdot 10^{6}$ $\pm$ $3.1 \cdot 10^{2}$ & $6.179273 \cdot 10^{9}$ $\pm$ $1.8 \cdot 10^{5}$ & 2 \\
9 & $1.10 \cdot 10^{10}$ & $9.80 \cdot 10^{8}$ & $1.045110 \cdot 10^{8}$ $\pm$ $6.2 \cdot 10^{3}$ & $2.692812 \cdot 10^{11}$ $\pm$ $8.2 \cdot 10^{6}$ & 2 \\
10 & $1.10 \cdot 10^{10}$ & $9.33 \cdot 10^{8}$ & $1.889201 \cdot 10^{9}$ $\pm$ $1.3 \cdot 10^{5}$ & $1.241497 \cdot 10^{13}$ $\pm$ $3.9 \cdot 10^{8}$ & 3 \\
11 & $1.10 \cdot 10^{10}$ & $8.96 \cdot 10^{8}$ & $3.566923 \cdot 10^{10}$ $\pm$ $2.8 \cdot 10^{6}$ & $6.026765 \cdot 10^{14}$ $\pm$ $1.9 \cdot 10^{10}$ & 4 \\
12 & $1.10 \cdot 10^{10}$ & $8.66 \cdot 10^{8}$ & $7.012027 \cdot 10^{11}$ $\pm$ $6.4 \cdot 10^{7}$ & $3.071324 \cdot 10^{16}$ $\pm$ $1.0 \cdot 10^{12}$ & 5 \\
13 & $1.10 \cdot 10^{10}$ & $8.44 \cdot 10^{8}$ & $1.431902 \cdot 10^{13}$ $\pm$ $1.5 \cdot 10^{9}$ & $1.638982 \cdot 10^{18}$ $\pm$ $5.4 \cdot 10^{13}$ & 6 \\
14 & $1.10 \cdot 10^{10}$ & $8.28 \cdot 10^{8}$ & $3.032472 \cdot 10^{14}$ $\pm$ $3.6 \cdot 10^{10}$ & $9.142727 \cdot 10^{19}$ $\pm$ $3.1 \cdot 10^{15}$ & 7 \\
15 & $3.11 \cdot 10^{11}$ & $2.31 \cdot 10^{10}$ & $6.655768 \cdot 10^{15}$ $\pm$ $2.4 \cdot 10^{10}$ & $5.323570 \cdot 10^{21}$ $\pm$ $3.4 \cdot 10^{16}$ & 249 \\
16 & $1.10 \cdot 10^{10}$ & $8.10 \cdot 10^{8}$ & $1.512467 \cdot 10^{17}$ $\pm$ $2.4 \cdot 10^{13}$ & $3.231993 \cdot 10^{23}$ $\pm$ $1.1 \cdot 10^{19}$ & 11 \\
17 & $1.10 \cdot 10^{10}$ & $8.08 \cdot 10^{8}$ & $3.552250 \cdot 10^{18}$ $\pm$ $6.3 \cdot 10^{14}$ & $2.044094 \cdot 10^{25}$ $\pm$ $6.9 \cdot 10^{20}$ & 12 \\
18 & $1.10 \cdot 10^{10}$ & $8.09 \cdot 10^{8}$ & $8.632116 \cdot 10^{19}$ $\pm$ $1.8 \cdot 10^{16}$ & $1.345581 \cdot 10^{27}$ $\pm$ $4.6 \cdot 10^{22}$ & 15 \\
19 & $1.10 \cdot 10^{10}$ & $8.12 \cdot 10^{8}$ & $2.167796 \cdot 10^{21}$ $\pm$ $4.9 \cdot 10^{17}$ & $9.211519 \cdot 10^{28}$ $\pm$ $3.1 \cdot 10^{24}$ & 19 \\
20 & $3.00 \cdot 10^{11}$ & $2.23 \cdot 10^{10}$ & $5.624473 \cdot 10^{22}$ $\pm$ $4.0 \cdot 10^{17}$ & $6.551806 \cdot 10^{30}$ $\pm$ $4.2 \cdot 10^{25}$ & 621 \\
25 & $8.30 \cdot 10^{10}$ & $6.50 \cdot 10^{9}$ & $1.066295 \cdot 10^{30}$ $\pm$ $2.9 \cdot 10^{25}$ & $2.060052 \cdot 10^{40}$ $\pm$ $2.5 \cdot 10^{35}$ & 824 \\
30 & $1.00 \cdot 10^{11}$ & $8.26 \cdot 10^{9}$ & $4.290822 \cdot 10^{37}$ $\pm$ $1.8 \cdot 10^{33}$ & $1.486361 \cdot 10^{50}$ $\pm$ $1.6 \cdot 10^{45}$ & 2032 \\
40 & $1.00 \cdot 10^{10}$ & $8.86 \cdot 10^{8}$ & $4.946806 \cdot 10^{53}$ $\pm$ $1.6 \cdot 10^{50}$ & $6.283492 \cdot 10^{70}$ $\pm$ $2.0 \cdot 10^{66}$ & 638 \\
50 & $1.00 \cdot 10^{10}$ & $9.26 \cdot 10^{8}$ & $5.054951 \cdot 10^{70}$ $\pm$ $3.8 \cdot 10^{67}$ & $2.625921 \cdot 10^{92}$ $\pm$ $8.2 \cdot 10^{87}$ & 725 \\
\bottomrule
\end{tabular}
\caption{Primitive $\beta$ function computation in $\phi^4$ theory.}
\label{tab:phi4_loop_values}

\end{table}

The positive modifications of Algorithms~\ref{algo:1PI} and \ref{algo:beaded} are implemented in the \texttt{C++} code that is attached to the paper. In fact, the positive version of the algorithm is equivalent to the original one up to a check of positivity of $\omega^k_D(L,n)$ in the recursion step~\eqref{eq:Zrec}. The results of a single pilot-study computation on the Perimeter Institute's computing cluster are listed in Table~\ref{tab:phi4_loop_values}, which contains the number of samples, the number of samples where $G$ was primitive $N_{\mathrm{prim}}$, $ \beta^{\mathrm{prim}}_{L+1}$, the runtime of the computation in hours and the value of 
$\beta H^{\mathrm{prim}}_{L+1}$, which is the Hepp-version of the primitive $\beta$ function in $\phi^4$ theory:
\begin{align*} \beta H^{\mathrm{prim}}_{L+1} = 2 Z^{4,\mathrm{pos}}_{4}(L,4) \int_{\mathcal M_{L,4}^4} \Theta( G \text{ is primitive} )\, \mu^{\tr,\mathrm{pos}}_{L,4} \, . \end{align*}
The main bottleneck of the computation is to check if a graph is primitive or not. This check can be performed in polynomial time by enumerating all $4$-edge cuts of the graph and checking if the respective cut splits the graph into two non-trivial components. However, the number of potential cuts grows as $|E_G|^4$; faster than the runtime of any other component of the computation.

The timings in Table~\ref{tab:phi4_loop_values} are not fully representative, as two different types of computers, one with an \texttt{AMD EPYC 7532 32-Core Processor} and one with an \texttt{Intel Xeon Gold 6148 CPU} were used for the computation. 
Again, we did not record the memory usage, as the core memory requirements of the algorithm would still be on the order of 200 kilobytes at loop order 100.
 The overall memory consumption is therefore dominated by operating system overhead.

The integrand $f(G,\bb x) = \left( { \mathcal U_G^\tr(\bb x)}/{\mathcal U_G(\bb x)} \right)^2 $ is much better-behaved than the integrand in the super-renormalizable case discussed above in Section~\ref{sec:applicationphi3}. As expected, the errors grow more slowly with the loop order than before in our $\phi^3$ theory example (see~Table~\ref{tab:phi4_loop_values}). We obtained estimates with a reasonable statistical error bound up to loop order 50. At loop orders 18 and below, the numbers confirm and refine results from \cite{Balduf:2023ilc,Balduf:2024gvv,Borinsky:2025ncs}. We emphasize that the new computations presented here were carried out in a fraction of the time required for the more basic, uniform graph sampling approach of these papers.

However, the statistical errors still seem to grow exponentially with the loop order. This growth can be explained as follows. The integrand, $f(G,\bb x)$, is always smaller than $1$, and it can become arbitrarily small at large enough loop order. The crude lower bound, ${2^{-2|E_G|}} \leq f(G,\bb x)$ obtained by estimating the number of spanning trees with the total number of subgraphs, can only be improved to a better, but still exponentially 
decreasing, bound (see~\cite{mckay1983spanning}). Hence,
the discussion in Section~\ref{sec:error} shows that we end up with an exponentially growing bound on the relative accuracy of our estimate $\delta[f] \leq 2^{4L}$, where we used that a $\phi^4$ theory $4$-point graph with $L$ loops has $2L$ edges. 
We therefore proved that we have an algorithm 
with which we can compute the primitive $\phi^4$ theory $\beta$ function at $L$ loops up to $\delta$ relative accuracy in runtime $\bigO(L^4 \delta^{-C L})$, where $C$ is a constant $> 0$. So, requiring fixed relative accuracy  still results in an exponential time algorithm in $L$.
As before, this bound and estimate are overly pessimistic, but in line with the empirical observation that the errors grow exponentially. We hope that systematic improvements of the global sampling algorithm will lead to a polynomial-time algorithm  for the full estimation of the $L$-loop perturbative coefficient at fixed relative accuracy.

\section{Perspectives \& generalizations}
\label{sec:perspective}
In this section, we summarize several perspectives and possible generalizations that emerged from the open directions highlighted throughout the article.

\begin{itemize}

\item 

Except for the brief excursions in Sections~\ref{sec:crit} and \ref{sec:renormalization}, this paper focused exclusively on \emph{bare tropicalized quantum field theory}. An exciting direction for future work is to investigate the renormalization group properties of tropicalized quantum field theory itself (see \cite{BaldufPanzer}).

Furthermore, it would be extremely valuable to incorporate renormalization directly into the sampling algorithm described in Section~\ref{sec:algo}, to compute renormalized $n$-point correlation functions and other quantities in situ. Given the numerical nature of the method, divergences would need to be addressed locally at the integrand level, so that only finite quantities remain to be integrated. Recent progress in local subtraction schemes (see \cite{Brown:2011pj,Hillman:2023ezp,Salvatori:2024nva}) should be particularly helpful while designing a fully renormalized algorithm.

\item

There are many ways to generalize the underlying quantum field theory beyond massive scalar fields. Natural first steps include massless scalar quantum field theories and theories with multiple scalar fields. Further interesting directions would be the tropicalization of quantum field theories with tensor structures, and eventually the tropicalization  of gauge theories and theories with fermions.

\item

For applications of the global tropical sampling algorithm from Section~\ref{sec:algo} to collider physics, it would be interesting to combine the algorithm from the present paper with that of \cite{Borinsky:2025asc} to directly sample pairs of a graph and the associated loop momenta along its edges. Such a combination could also prove useful for the evaluation of phase-space integrals and simulating particle scattering processes.

\item
Conceptually, the sampling method over the moduli space of graphs could provide a way to study the number-theoretic properties of Feynman integrals (see~\cite{Brown:2015fyf}) quantitatively. Interesting questions include, for example, what proportion of the contributing integrals can be expressed in terms of generalized hyperlogarithms, and which transcendental weights dominate the contributions.

\item 

Furthermore, it is plausible that the algorithm could be significantly improved using systematic variance-reduction techniques. In particular, it would be interesting to achieve a rigorous reduction of the variance such that the Monte Carlo application of the algorithm maintains polynomial complexity in the loop order, multiplicity, and target accuracy.

\end{itemize}

\providecommand{\href}[2]{#2}\begingroup\raggedright\endgroup

\end{document}